\documentclass[11pt]{article}
\usepackage[hmargin=1in,vmargin=1in]{geometry}
\usepackage[dvipsnames,usenames,table]{xcolor}
\usepackage{makecell} 
\usepackage{hyperref}
\hypersetup{colorlinks=true, urlcolor=Blue, citecolor=Green, linkcolor=BrickRed, breaklinks, unicode}

\usepackage{hchang}
\usepackage{thm-restate}
\usepackage{cleveref}
\usepackage{comment}
\usepackage{soul}

\makeatletter
\renewcommand*{\@fnsymbol}[1]{\ensuremath{\ifcase#1\or *\or \dagger\or \ddagger\or
   \mathsection\or \mathparagraph\or \|\or **\or \dagger\dagger
   \or \ddagger\ddagger \else\@ctrerr\fi}}
\makeatother

\newtheorem{theorem}{Theorem}[section]
\newtheorem{lemma}[theorem]{Lemma}
\newtheorem{claim}[theorem]{Claim}
\newtheorem{observation}[theorem]{Observation}
\newtheorem{problem}[theorem]{Problem}
\newtheorem{remark}[theorem]{Remark}
\newtheorem{corollary}[theorem]{Corollary}

\newtheorem{definition}[theorem]{Definition}
\newtheorem{question}[theorem]{Question}
\newtheorem{invariant}[theorem]{Invariant}

\newcommand{\len}[1]{{\norm{#1}}}
\def\eps{\e}
\newcommand{\weight}{w}
\DeclareMathOperator\dist{\delta}
\DeclareMathOperator{\cost}{cost}
\DeclareMathOperator{\tw}{tw}
\DeclareMathOperator{\level}{\mathit{lv}}
\newcommand{\dom}{\ensuremath{\mathrm{dom}}}
\newcommand{\bag}[1]{\ensuremath{B_{#1}}}
\newcommand{\bdry}{\partial\!}
\newcommand{\buff}{\mathcal{N}}
\newcommand{\seen}[1]{\ensuremath{\cS|_{#1}}}
\newcommand{\seenhat}[1]{\ensuremath{\hat{\cS}|_{#1}}}

\def\cX{\ensuremath{\mathcal{X}}}
\def\cS{\ensuremath{\mathcal{S}}}
\def\cT{\ensuremath{\mathcal{T}}}
\def\cC{\ensuremath{\mathcal{C}}}
\def\cF{\ensuremath{\mathcal{F}}}

\begin{document}

\begin{titlepage}

\title{Shortcut Partitions in Minor-Free Graphs: \\Steiner Point Removal, Distance Oracles, Tree Covers, and More}

\author{%
Hsien-Chih Chang%
\thanks{Department of Computer Science, Dartmouth College. Email: {\tt hsien-chih.chang@dartmouth.edu}.} 
\and 
Jonathan Conroy%
\thanks{Department of Computer Science, Dartmouth College. Email: {\tt jonathan.conroy.gr@dartmouth.edu}}  
\and 
Hung Le%
\thanks{Manning CICS, UMass Amherst. Email: {\tt hungle@cs.umass.edu}}  
\and
Lazar Milenkovic%
\thanks{Tel Aviv University}  
\and
Shay Solomon%
\thanks{Tel Aviv University}  
\and
Cuong Than%
\thanks{Manning CICS, UMass Amherst. Email: {\tt cthan@cs.umass.edu}}  
}

\date{\today}

\maketitle
	
\thispagestyle{empty}

\begin{abstract}

The notion of \emph{shortcut partition}, introduced recently by Chang, Conroy, Le, Milenković, Solomon, and Than~\cite{CCLMST23}, is a new type of graph partition into low-diameter clusters.  
Roughly speaking, the shortcut partition guarantees that for every two vertices $u$ and $v$ in the graph, there exists a path
between $u$ and $v$
that intersects only a few clusters. 
They proved that any planar graph admits a shortcut partition and gave several applications, including a construction of tree cover for arbitrary planar graphs with stretch $1+\e$ and $O(1)$ many trees for any fixed $\e\in (0,1)$. 
However, the construction heavily exploits planarity in multiple steps,
and is thus inherently limited to planar~graphs.

In this work, we breach the ``planarity barrier'' to construct a shortcut partition for $K_r$-minor-free graphs for any $r$. To this end, we take a completely different approach 
---  our key contribution is a novel deterministic variant of the \emph{cop decomposition} in minor-free graphs \cite{andreae1986pursuit,agg14}. 
Our shortcut partition for $K_r$-minor-free graphs yields several direct applications.
Most notably, we construct the first \emph{optimal} distance oracle for $K_r$-minor-free graphs, with $1+\e$ stretch, linear space, and constant query time for any fixed $\e \in (0,1)$. 
The previous best distance oracle~\cite{AG06} uses $O(n\log n)$ space and $O(\log n)$ query time, and its construction relies on Robertson-Seymour structural theorem
and other sophisticated tools. 
We also obtain the first tree cover of $O(1)$ size for minor-free graphs with stretch $1+\eps$, while the previous best $(1+\eps)$-tree cover has size $O(\log^2 n)$ \cite{BFN19Ramsey}.  

As a highlight of our work, we employ our shortcut partition to resolve a major open problem
--- the \emph{Steiner point removal (SPR)} problem: Given any set $K$ of \emph{terminals} in an arbitrary edge-weighted planar graph $G$, is it possible to construct a minor $M$ of $G$ whose vertex set is~$K$, which preserves the shortest-path distances between all pairs of terminals in $G$ up to a \emph{constant} factor?  
Positive answers to the SPR problem were only known for very restricted classes of planar graphs: trees~\cite{Gupta01}, outerplanar graphs~\cite{BG08}, and series-parallel graphs~\cite{HL22}.   
We resolve the SPR problem in the affirmative for any planar graph, and more generally for any $K_r$-minor-free graph for any fixed $r$. 
To achieve this result, we prove the following general reduction and combine it with our new shortcut partition: 
For any graph family closed under taking subgraphs, the existence of a shortcut partition yields a positive solution to the SPR problem. 
\end{abstract}

\end{titlepage}

\section{Introduction}
\label{sec:Intro}

Partitioning a graph into clusters is a fundamental primitive for designing algorithms. 
Perhaps the most basic requirement of a partition is that every cluster would have a small diameter. 
However, to be useful, most partitions require one or more additional constraints, and achieving these constraints is the key to the power of those partitions.  
For example, \emph{probabilistic partition}~\cite{Bar96}, a principal tool in the metric embedding literature, guarantees that the probability of any two vertices being placed into different clusters is proportional to their distance. 
\emph{Sparse partition}~\cite{AP90,JLNR05} guarantees that each cluster has neighbors in only a few other clusters. 
\emph{Scattering partition}~\cite{Filtser20B} guarantees that each shortest path up to a certain length
only intersects a small number of clusters.  
These partitions have found a plethora of applications in a wide variety of areas, such as metric embeddings, distributed computing, routing, and algorithms for network design problems, to name a few. 

Recently, Chang, Conroy, Le, Milenković, Solomon, and Than~\cite{CCLMST23} introduced a new notion of partition called \emph{shortcut partition}.
Roughly speaking, a shortcut partition guarantees that for every two vertices $u$ and $v$ in the graph, there exists a low-hop path in the cluster graph between $C_u$ and $C_v$, where $C_u$ and $C_v$ are the clusters containing $u$ and $v$, respectively. 
More formally, a \EMPH{clustering} of a graph $G$ is a partition of the vertices of $G$ into connected \EMPH{clusters}.   
The \EMPH{cluster graph $\check{G}$} of a clustering $\mathcal{C}$ of $G$ is the graph where each vertex of $\check{G}$ corresponds to a cluster in $\mathcal{C}$, and there is an edge between two vertices in~$\check{G}$ if there is an edge in $G$ whose endpoints are in the two corresponding clusters. 

\begin{definition} \label{def:shortcut-partition}
An \EMPH{$(\e,h)$-shortcut partition} is a clustering $\mathcal{C} = \set{C_1, \ldots, C_m}$ of $G$ such that:
\begin{itemize}
    \item \textnormal{[Diameter.]}  the \emph{strong}%
    \footnote{The \EMPH{strong} diameter of cluster $C$ is the one of induced subgraph $G[C]$. 
    In contrast, the \EMPH{weak} diameter of $C$ is $\max_{u,v\in C} \dist_G(u,v)$. Here, and throughout this paper, \EMPH{$\dist_G(u,v)$} denotes the distance between $u$ and $v$ in graph $G$.}
    diameter of each cluster $C_i$ is at most $\e \cdot \diam(G)$;

    \item \textnormal{[Low-hop.]}  for any vertices $u$ and $v$ in $G$,
    there is a path $\check{\pi}$ in the cluster graph $\check{G}$ between the clusters containing $u$ and $v$ such that:
    \begin{enumerate}
        \item[(1)]  $\check{\pi}$ has hop-length at most $\e h\cdot \Ceil{ \frac{\dist_G(u,v)}{\eps \cdot \diam(G)} }$, 

        \item[(2)]   $\check{\pi}$ only contains a subset of clusters that have nontrivial intersections with a given shortest path in $G$ between $u$ and $v$.
    \end{enumerate}
\end{itemize}
The hop-length of $\check{\pi}$ is measured in \emph{units of $\Delta \coloneqq\eps \cdot \diam(G)$}; if $u$ and $v$ have distance at most $\alpha\cdot\Delta$, then the hop-length of $\check{\pi}$ is $\alpha\cdot\e h$, where $\alpha$ ranges between $0$ and $1/\e$.
In particular, the hop-length of $\check{\pi}$ is always at most $O(h)$ as $\dist_G(u,v)$ is at most $\diam(G)$.
\end{definition}

The shortcut partition is similar to the scattering partition introduced by Filtser~\cite{Filtser20B}. 
A key difference is that in a scattering partition, \emph{every} shortest path of length $\alpha\e \cdot \diam(G)$ intersects at most $O(\alpha)$ clusters, while in a shortcut partition, it only requires that there is a low-hop path in the cluster graph between the two clusters containing the path's endpoints. 
The fact that scattering partition requires a stronger guarantee on shortest paths makes it very difficult to construct; it remains an open problem whether scattering partition for every planar graph exists~\cite[Conjecture~1]{Filtser20B}.  
Although shortcut partition provides a weaker guarantee, 
it is already sufficient for many applications as shown in previous work~\cite{CCLMST23}, including the first tree cover in planar graphs with stretch $1+\e$ using $O(1)$ many trees for any fixed $\e\in (0,1)$, a simpler proof to the existence of a $+\e\cdot\diam(G)$ additive embedding of planar graph into bounded-treewidth graph, distance oracles, labeling schemes, (hop-)emulators, and more. 

For any given $\e\in (0,1)$, the authors of \cite{CCLMST23} constructed an $(\e,O(\e^{-2}))$-shortcut partition for any planar graph. 
This naturally motivates the question of constructing a shortcut partition for broader classes of graphs, specifically $K_r$-minor-free graphs.%
\footnote{We sometimes drop the prefix ``$K_r$-'' in $K_r$-minor-free graphs when the clique minor has constant size $r$.}
This will open the door to seamlessly extend algorithmic results from planar graphs to $K_r$-minor-free graphs. 
However, the construction of
\cite{CCLMST23}
heavily exploits planarity in multiple steps. It starts from the outerface of $G$, and works toward the interior of $G$ in a recursive manner, similar in spirit to Busch, LaFortune, and Tirthapura~\cite{BLT14}.  
Specifically, the construction first finds a collection of subgraphs of $G$ call \emph{columns} such that every vertex near the outerface of $G$ belongs to one of the columns.
The construction then recurs on subgraphs induced by vertices that are not in any of the columns.  
The overall construction produces a structure called the \emph{grid-tree hierarchy}, which is then used to construct a shortcut partition.  
The construction relies on the fact that each column contains a shortest path between two vertices on the outer face, which splits the graph into two subgraphs using Jordan curve theorem.
As a result, constructing a shortcut partition for $K_r$-minor-free graphs requires breaking away from the planarity-exploiting framework of \cite{CCLMST23}. 

\medskip
In this work, we overcome this barrier and construct a shortcut partition for $K_r$-minor-free graphs. 

\begin{restatable}{theorem}{ShortcutMinor}
\label{thm:shortcut-minor} 
Any edge-weighted $K_r$-minor-free graph admits an $(\e,2^{O(r \log r)}/\e)$-shortcut partition for any $\eps\in (0,1)$.
\end{restatable}

\noindent\textbf{Remark.} \,\,
{\sl
\Cref{def:shortcut-partition} is slightly stronger than the corresponding definition of shortcut partition for planar graphs in \cite{CCLMST23} (Definition 2.1 in their paper). Specifically, their definition states that the hop-length of $\check{\pi}$ is at most $h$, regardless of $\dist_G(u,v)$, while our definition allows smaller hop-lengths for smaller distances. 
(For example, when $\dist_G(u,v) = \e \cdot \diam(G)$, the hop-length of $\check{\pi}$ is $O(\eps h)$ instead of $h$.)
Another difference is that, in the current definition, the nontrivial intersections of clusters
contained by $\check{\pi}$ 
stated in condition (2) of the ``low-hop'' property are with respect to a shortest path in the graph, whereas in \cite{CCLMST23} they are with respect to an approximate shortest path; 
we discuss this point further in \Cref{sec:other-apps}.
In particular, the shortcut partition provided by 
\Cref{thm:shortcut-minor} for minor-free graphs subsumes the  one in \cite{CCLMST23} for planar graphs. 

The hop length of $2^{O(r \log r)}/\e$ of the shortcut partition in \Cref{thm:shortcut-minor} is \emph{optimal} for every constant~$r$ up to a constant factor: any shortcut partition of a path
would have hop length $1/\eps$ between the two endpoints of the path. 
Also, in the particular case of planar graphs, our shortcut partition in fact improves over \cite{CCLMST23}; the hop length of their partition is $O(\eps^{-2})$.
}

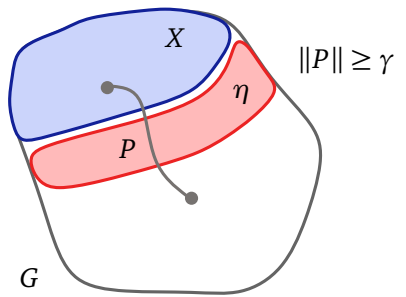
\begin{wrapfigure}{r}{0.33\textwidth}
\centering
    \def\svgwidth{0.5\textwidth}
        
\begingroup%
  \makeatletter%
  \providecommand\color[2][]{%
    \errmessage{(Inkscape) Color is used for the text in Inkscape, but the package 'color.sty' is not loaded}%
    \renewcommand\color[2][]{}%
  }%
  \providecommand\transparent[1]{%
    \errmessage{(Inkscape) Transparency is used (non-zero) for the text in Inkscape, but the package 'transparent.sty' is not loaded}%
    \renewcommand\transparent[1]{}%
  }%
  \providecommand\rotatebox[2]{#2}%
  \newcommand*\fsize{\dimexpr\f@size pt\relax}%
  \newcommand*\lineheight[1]{\fontsize{\fsize}{#1\fsize}\selectfont}%
  \ifx\svgwidth\undefined%
    \setlength{\unitlength}{566.6554372bp}%
    \ifx\svgscale\undefined%
      \relax%
    \else%
      \setlength{\unitlength}{\unitlength * \real{\svgscale}}%
    \fi%
  \else%
    \setlength{\unitlength}{\svgwidth}%
  \fi%
  \global\let\svgwidth\undefined%
  \global\let\svgscale\undefined%
  \makeatother%
  \begin{picture}(1,0.46294149)%
    \lineheight{1}%
    \setlength\tabcolsep{0pt}%
    \put(0,0){\includegraphics[width=\unitlength,page=1]{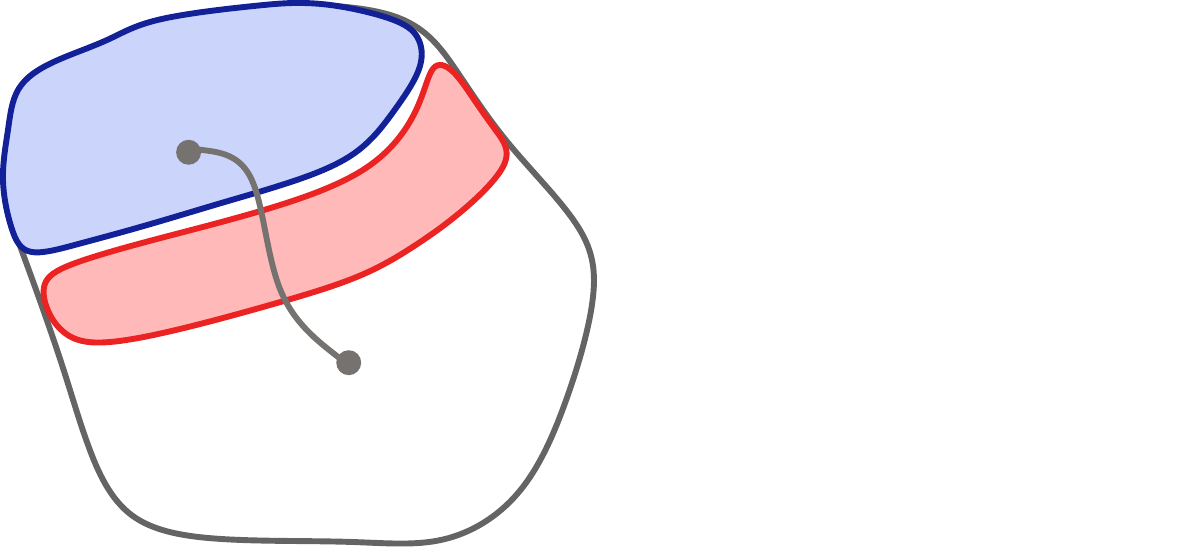}}%
    \put(0.18051346,0.22092152){\makebox(0,0)[lt]{\lineheight{1.25}\smash{\begin{tabular}[t]{l}{\textit{P}}\end{tabular}}}}%
    \put(0.46657258,0.36517687){\makebox(0,0)[lt]{\lineheight{1.25}\smash{\begin{tabular}[t]{l}{\textit{$\len{P}\ge\gamma$}}\textbf{ }\end{tabular}}}}%
    \put(0.0188658,0.01325153){\makebox(0,0)[lt]{\lineheight{1.25}\smash{\begin{tabular}[t]{l}{\textit{G}}\end{tabular}}}}%
    \put(0.25621727,0.39932323){\makebox(0,0)[lt]{\lineheight{1.25}\smash{\begin{tabular}[t]{l}{\textit{X}}\end{tabular}}}}%
    \put(0.36098438,0.31692986){\makebox(0,0)[lt]{\lineheight{1.25}\smash{\begin{tabular}[t]{l}{\textit{$\eta$}}\end{tabular}}}}%
  \end{picture}%
\endgroup%

    \caption{A cluster $X$ ``cut off'' by $\eta$ from part of graph $G$. 
    There is a buffer of width~$\gamma$ between $X$ and the part of the graph that it is cut off from.}
    \label{fig:cut-off}
\end{wrapfigure} 

\paragraph{Techniques.}
We base our construction on a modified \emph{cop decompo-\allowbreak sition} for minor-free graphs, first introduced by Andreas~\cite{andreae1986pursuit} in the context of the \emph{cops-and-robbers} game. Loosely speaking, a cop decomposition is a rooted tree decomposition where vertices of every bag belong to at most $r-2$ single-source shortest path (SSSP) trees, called \emph{skeletons}. 
Abraham, Gavoille, Gupta, Neiman, and Talwar~\cite{agg14}, in their construction of a padded decomposition for $K_r$-minor-free graphs, adapted the cop decomposition by allowing each bag to contain up to $r-2$ clusters%
\footnote{Later in the paper, we rename the clusters as \emph{supernodes}; we reserve the former term for the actual clusters in the shortcut partition to be constructed.}, each of which is the set of vertices within radius $\sigma \cdot \Delta$ from a skeleton in the bag; here $\sigma$ is a parameter in $(0,1)$ randomly sampled from a truncated exponential distribution, and $\Delta \coloneqq \e \cdot \diam(G)$.  
One of their goals is to ensure that for every vertex $v$ in the graph, the process of constructing the cop decomposition guarantees that the number of (random) clusters that could contain $v$ is small \emph{in expectation}; this is the bulk of their analysis,  through clever and sophisticated  use of potential function and setting up a (sub)martingale. Their result can be interpreted as guaranteeing what we call a \EMPH{buffer property}\footnote{There is a technical difference between our buffer property and that of \cite{AGGM06}, which we will clarify in \Cref{SS:buffered-cop}.}:

\begin{quote}
If one cluster $X$ is ``cut off'' from a piece of the graph by another cluster, then any path from~$X$ to that piece has length at least $\gamma$, which we call the \emph{buffer width}. 
\end{quote}

\noindent In particular, the construction of \cite{agg14} intuitively implies that the \emph{expected} buffer width is about~$\gamma$. (Their end goal is a stochastic partition, and hence they could afford the buffer property in expectation.) Their expected bound on the buffer width is insufficient for our shortcut partition, as well as for all other applications considered in our paper. One either has to guarantee that the buffer width holds with high probability or, ideally, holds deterministically. 

In this work, we achieve the buffer property \emph{deterministically}. 
To do this, we add a layer of recursion on top of the cop decomposition by~\cite{agg14} to directly fix the buffer property whenever it is violated, thereby bypassing the need for the complicated analysis of the potential function. In more detail, we build a cop decomposition by iteratively creating clusters.  
At each point in the construction, we create an SSSP tree connecting to some existing clusters,
and initialize a new cluster with that tree as the cluster's skeleton. 
Our idea to enforce the buffer property is natural: we (recursively) assign those vertices that violate the property to join previously-created clusters. 
Specifically, whenever a cluster $X$ is cut off by some cluster $\eta$, we assign every vertex within distance $\gamma$ of $X$ to be a part of some existing clusters. 
However, enforcing the buffer property directly comes at the cost of increasing the radius of some existing clusters --- recall that we want all points in a cluster to be at most $O(\Delta)$ distance away from its skeleton. Therefore, our implementation of vertex assignment is very delicate; otherwise, the diameter of a cluster could continue growing, passing the diameter bound prescribed by the shortcut partition. Our key insight is to show that during the course of our construction,  each cluster can only be expanded a single time for each of the $O(r)$ clusters that it can ``see''.
This lets us achieve a deterministic buffer width of $\gamma = O(\Delta/r)$.

\subsection{Steiner Point Removal Problem}

In the Steiner Point Removal (SPR) problem, we are given an undirected weighted graph $G = (V,E,\weight)$ with vertex set $V$, edge set $E$,  nonnegative weight function $\weight$ over the edges, and a subset $K$ of $V$. The vertices in $K$ are called \EMPH{terminals} and the vertices in $V \setminus K$ are called \EMPH{non-terminal} or \EMPH{Steiner} vertices. 
The goal in the SPR problem is to find a graph \emph{minor} $M$ of $G$ such that $V(M) = K$, and for every pair $t_1, t_2$ 
of terminals in $K$, $\dist_M(t_1, t_2) \leq \alpha \cdot\dist_G(t_1, t_2)$, for some {constant} $\alpha \geq 1$; such a graph minor $M$ of $G$ is called a \EMPH{distance-preserving minor} of $G$ with \EMPH{distortion~$\alpha$}.%
\footnote{In the literature~\cite{KNZ14} the term \emph{distance-preserving minor} allows the existence of Steiner vertices as well, with the goal to minimize their usage.  For our purpose we do not allow any Steiner vertices.} 

Gupta~\cite{Gupta01} was the first to consider the problem of removing Steiner points to preserve all terminal distances. 
He showed that every (weighted) \emph{tree} can be replaced by another tree on the terminals where the shortest path distances are preserved up to a factor of 8;
another proof of this result is given by~\cite{FKT19}. 
Chan, Xia, Konjevod, and Richa~\cite{CXKR06} observed that Gupta's construction in fact produces a distance-preserving minor of the input tree,
and showed a matching lower bound: 
there exists a tree and a set of terminals, such that any distance-preserving minor of that tree must have distortion at least $8(1-o(1))$. 
Both Chan \etal~\cite{CXKR06} and Basu and Gupta~\cite{BG08} considered the following question:

\begin{question} \label{q:main}
Does every $K_r$-minor-free graph for any fixed $r$ admit a distance-preserving minor with \emph{constant distortion}?
\end{question} 

Question~\ref{q:main} has attracted significant research attention over the years, and numerous works have attempted to attack it from different angles. 
Some introduced new frameworks~\cite{FKT19,Filtser19,Filtser20} that simplify known results; others considered the problem for general graphs, establishing the distortion bound of $O(\log |K|)$ after a sequence of works~\cite{KNZ15,Cheung2018,Filtser19}; there are also variants where Steiner points are allowed, but their number should be minimized~\cite{KNZ14,CGH16,CRT22}; and yet another achieved a constant \emph{expected} distortion~\cite{EGKRTT14}.

Nevertheless, Question~\ref{q:main} remains widely open: a positive solution for $K_r$-minor-free graphs is not known for any $r\geq 5$. Gupta's result for trees~\cite{Gupta01} can be seen as providing a solution for $K_3$-minor-free graphs.  
Basu and Gupta~\cite{BG08} gave a positive answer for outerplanar graphs (which is $(K_{2,3},K_4)$-minor-free). 
Recently, Hershkowitz and Li~\cite{HL22} provided a solution for $K_4$-minor-free graphs, also known as  \emph{series-parallel graphs}.  
Even for planar graphs, a subclass of $K_5$-minor-free graphs, the answer is not known. Both outerplanar and series-parallel graphs are very restricted classes of planar graphs: they have treewidth at most $2$. For slightly larger graph classes, such as treewidth-$3$ planar graphs or $k$-outerplanar graphs for any constant $k$, the SPR problem has remained open to date.  

\medskip
We resolve  Question~\ref{q:main} in the affirmative, thus solving the SPR problem for minor-free graphs in its full generality:
\begin{theorem}
    \label{thm:spr-sol}
    Let $G = (V, E, \weight)$ be an arbitrary edge-weighted $K_r$-minor-free graph and let $K \subseteq V$ be an arbitrary set of terminals. 
    Then, there is a solution to the SPR problem on $G$ with distortion $2^{O(r \log r)}$.
\end{theorem}

We prove \Cref{thm:spr-sol} by devising a general reduction from SPR to shortcut partition. Specifically:

\begin{theorem}
\label{thm:shortcut-to-SPR} 
If every subgraph of $G$ admits an $(\e,f(r)/\e)$-shortcut partition for every $\e \in (0,1)$,
then $G$ admits a solution to the SPR problem with distortion $O(f(r)^{13})$.    
\end{theorem}

The proof of \Cref{thm:shortcut-to-SPR} builds on a reduction by Filtser~\cite{Filtser20B}, from the SPR problem to that of finding \emph{scattering partitions}, which require every shortest path between two vertices to intersect only a small number of clusters. 
We introduce an inherently relaxed notion which we call the \emph{approximate scattering partition} (Definition~\ref{def:scatter}) ---
which among other changes uses \emph{approximate} shortest paths rather than exact shortest paths ---
and adapt Filtser's reduction using the new notion. 
The first challenge underlying this adaptation is that, unlike shortest paths, an approximate shortest path does not have the optimal substructure property (any subpath of a shortest path is also a shortest path).  
The second and perhaps more significant challenge stems from the fact that the partition only guarantees the existence of \emph{some} low-hop path in the cluster graph, and the distortion to its length is \emph{not} with respect to  the distance between the two endpoints.
We explain the differences in detail in Section \ref{sec:SPR-from-Shortcut}.
Consequently, we have to make some crucial changes in the reduction, and more so in its analysis.

\medskip \noindent
We observe that \Cref{thm:shortcut-to-SPR} together with a shortcut partition in \Cref{thm:shortcut-minor} gives us a solution to the SPR problem with $O(1)$ distortion in $K_r$-minor-free graphs, since in this case, $f(r) = r^{O(r)}$ and $r$ is fixed.

\subsection{Other Applications of Our Results}\label{subsec:app-intro}

\paragraph{Distance oracle.} 
An \EMPH{$\alpha$-approximate distance oracle} is a compact data structure for graph $G$ that given any two vertices $u$ and $v$, return the distance between $u$ and $v$ in $G$ up to a factor of $\alpha$. 
In constructing a distance oracle, we would like to minimize the \emph{distortion} parameter $\alpha$, the \emph{space} usage, and the \emph{time} it takes to answer a query; there is often a tradeoff between the three parameters. 

Constructing $(1+\e)$-approximate distance oracles for planar graphs has been extensively studied. A long line of work~\cite{Thorup04,Klein02,KKS11,WulffNilsen16,GX19,CS19} recently culminated in an optimal distance oracle with linear space and constant query time by  Le and Wulff-Nilsen~\cite{LW21}. On the other hand, the only known $(1+\e)$-approximate distance oracle for $K_r$-minor-free graphs achieving $O(n\log n)$ space and $O(\log n)$ query time (for any constant $\e\in (0,1)$ and constant $r$) was by Abraham and Gavoille~\cite{AG06}. 
The main reason is that the topology of $K_r$-minor-free graphs is much more complicated, and many techniques from planar graphs --- such as reduction to additive distance oracles~\cite{KKS11,LW21} or more sophisticated use of planar shortest path separators~\cite{WulffNilsen16} --- do not extend to $K_r$-minor-free graphs. 
Even the shortest path separator~\cite{AG06} in $K_r$-minor-free graphs does not behave as well as its planar counterpart~\cite{GKR01,Thorup04}: each path in the separator in $K_r$-minor-free graphs is not a shortest path of the \emph{input graph}, but a shortest path of its \emph{subgraph} after some previous paths were removed. 
As a result, despite significant recent progress on approximate distance oracles for planar graphs, the following problem remains open:

\begin{problem}\label{problem:dist-oracle}
    Design a $(1+\e)$-approximate distance oracle for $K_r$-minor-free graphs with linear space and constant query time for fixed $\e$ and $r$.
\end{problem}

In this work, we resolve \Cref{thm:app-distance-oracle} affirmatively. Our oracle can also be implemented in the pointer-machine model, matching the best-known results for planar graphs~\cite{CCLMST23}. 

\begin{theorem} \label{thm:app-distance-oracle} Given any parameter $\eps \in (0,1)$, and any edge-weighted undirected $K_r$-minor-free graphs with $n$ vertices, we can design a $(1+\eps)$-approximate distance oracle with the following guarantees:
\begin{itemize}\itemsep=0pt
    \item Our distance oracle has space $n \cdot 2^{r^{O(r)}/\e}$ and query time $2^{r^{O(r)}/\e}$ in the word RAM model with word size $\Omega(\log n)$. 
    Consequently, for fixed $\eps$ and $r$, the space is $O(n)$ and query time is $O(1)$. 
    \item Our distance oracle has space $O(n\cdot 2^{r^{O(r)}/\e})$ and query time $O(\log \log n\cdot 2^{r^{O(r)}/\e})$ in the pointer machine model. 
\end{itemize}
\end{theorem}

\noindent Our oracle is constructed via tree covers, which we will discuss next.

\paragraph{Tree cover.} 
An \EMPH{$\alpha$-tree cover $\mathcal{T}$} of a metric space $(X,\delta_X)$ for some $\alpha \geq 1$ is a collection of trees such that: 
(1) every tree $T\in \mathcal{T}$ has $X\subseteq V(T)$ and $d_T(x,y)\geq \delta_X(x,y)$ for every two points $x,y\in X$, and 
(2) for every two points $x,y\in X$, there exists a tree $T\in \mathcal{T}$ such that $d_{T}(x,y)\leq \alpha\cdot \delta_X(x,y)$. We call $\alpha$ the \emph{distortion} of the tree cover $\mathcal{T}$. The size of the tree cover is the number of trees in $\mathcal{T}$. 

Tree covers have been extensively studied for many different metric spaces~\cite{AP92,AKP94,ADMSS95,GKR01,BFN19Ramsey,FL22,KLMS22}.  
Gupta, Kumar, and Rastogi~\cite{GKR01} showed among other things that planar metrics admit a tree cover of distortion 3 and size $O(\log n)$. 
Bartal, Fandina, and Neiman~\cite{BFN19Ramsey} reduced the distortion to $1+\e$ for any fixed $\e \in (0,1)$ at the cost of a higher number of trees, $O(\log^2 n)$. 
Their result also holds for any $K_r$-minor-free graphs with a fixed $r$;
however, because of the usage of shortest path separator~\cite{AG06}, the final tree cover size contains a hidden dependency on $r$ which is the Robertson-Seymour constant~\cite{RS03}, known to be bigger than the tower function of $r$. 
Their work left several questions open: (a) Can we construct a $(1+\e)$ tree cover of $O(1)$ size for planar graphs, and more generally $K_r$-minor-free graphs? (b)  Can we avoid Robertson-Seymour decomposition and achieve a more practical construction? 

The shortcut partition introduced by Chang \etal~\cite{CCLMST23} partially resolved the first question: 
they constructed a $(1+\e)$-tree cover for \emph{planar graphs} of $O(1)$ size.  Using our new shortcut partition in \Cref{thm:shortcut-minor}, we resolve the question of Bartal \etal\ for all $K_r$-minor-free graphs. 
As our construction is rooted in the cop decomposition, the construction might behave reasonably well even when the graph is not strictly $K_r$-minor-free, as the performance ultimately depends on the width of the buffer and the number of times a cluster can expand.
This provides a more practical alternative to the Robertson-Seymour decomposition.  

\begin{theorem}\label{thm:tree-cover}  Let $G$ be any edge-weighted undirected $K_r$-minor-free graph with $n$ vertices.  For any parameter $\eps \in (0,1)$, there is a $(1+\eps)$-tree cover  for the shortest path metric of $G$ using $2^{r^{O(r)}/\eps}$~trees.
\end{theorem}

Given a tree cover $\mathcal{T}$ in \Cref{thm:tree-cover}, we can obtain a $(1+\e)$-approximate distance oracle in \Cref{thm:app-distance-oracle} as follows. 
The distance oracle consists of  $\mathcal{T}$ and an LCA data structure for each tree in $\mathcal{T}$. 
For each query pair $(u,v)$, we iterate through each tree, compute the distance on the tree using LCA data structure, and then return $\min_{T\in \mathcal{T}} d_T(u,v)$. 
The query time and space are as described in \Cref{thm:app-distance-oracle} because $|\mathcal{T}|=2^{r^{O(r)}/\eps}$; the distortion is $1+\eps$ since the distortion of the tree cover is $1+\e$.

\paragraph{Additive embeddings for apex-minor-free graphs.}  
Graph $A$ is an \EMPH{apex graph} if there exists a vertex $a\in V(A)$, called the \EMPH{apex}, such that $A\setminus \{a\}$ is a planar graph. 
A graph $G$ is \EMPH{apex-minor-free} if it excludes some apex graph $A$ of $O(1)$ size as a minor. 
We note that apex-minor-free graphs include planar graphs and, more generally, \emph{bounded-genus graphs} as subclasses.  
We show that our shortcut partition also gives the first deterministic additive embeddings of apex-minor-free graphs into bounded-treewidth graphs. 

Given a weighted graph $G$ of diameter $\Delta$, we say that a (deterministic) embedding $f: V(G)\rightarrow H$ of $G$ into $H$ has \EMPH{additive distortion} $+\e \Delta$
if $d_G(x,y)\leq d_H(f(x), f(y)) \leq d_G(x,y) + \e \Delta$ for every $x,y\in V(G)$. 
The goal is to construct an embedding $f$ such that the treewidth of $H$, denoted by $\tw(H)$, is minimized. 
Ideally, we would like $\tw(H)$ to depend only on $\eps$ and not on the number of vertices of $G$. 

Additive embeddings have been studied recently for planar graphs~\cite{FKS19,FL22,CCLMST23} and for minor-free graphs~\cite{CFKL20}. 
A key result in this line of work is an additive embedding for planar graphs where the treewidth of $H$ is polynomially dependent on $\e$~\cite{FKS19}; specifically, they achieved $\tw(H) = O(1/\e^{c})$ for some constant $c\geq 58$, which was recently improved to $\tw(H) = O(1/\e^{4})$~\cite{CCLMST23}. 
Cohen-Addad \etal~\cite{CFKL20} constructed a family of apex graphs and showed that any deterministic embedding with additive distortion $+\Delta/12$ ($\e = 1/12$) for the family must have treewidth $\Omega(\sqrt{n})$. 
Their result left an important question regarding additive embeddings of apex-minor-free graphs. 
Here we use the shortcut partition in \Cref{thm:shortcut-minor} to resolve this problem, and thereby completing our understanding of deterministic additive embeddings of graphs excluding a fixed minor into bounded-treewidth graphs. 

\begin{theorem}\label{thm:add-apex-minor}
    Let $G$ be any given edge-weighted graph of $n$ vertices excluding a fixed apex graph as a minor. 
    Let $\Delta$ be the diameter of $G$.  
    For any given parameter $\eps \in (0,1)$, we can construct in polynomial time a deterministic embedding of $G$ into a graph $H$ such that the additive distortion is $+\eps \Delta$ and $\tw(H) = 2^{O(\e^{-1})}$.
\end{theorem}

In addition to the aforementioned results, we also obtain generalizations to minor-free graphs of results from~\cite{CCLMST23}; 
we simply use our tree cover from \Cref{thm:tree-cover} in place of their tree cover theorem for planar graphs. 
The results include (1) the first $(1+\e)$-emulator of linear size for minor-free graphs, (2) low-hop emulators for minor-free metrics, (3) a compact distance labeling scheme for minor-free graphs, and (4) routing in minor-free metrics. 
We refer readers to~\cite{CCLMST23} for more details.

\paragraph{Organization.} 
In \Cref{sec:SPR-from-Shortcut} we resolve the SPR problem by constructing approximate scattering partition using the shortcut partition in \Cref{thm:shortcut-minor}. 
In \Cref{SS:buffered-cop}, we introduce and describe in full detail the construction of the buffered cop decomposition, which we will use in \Cref{sec:shortcut} to construct the shortcut partition. 
In \Cref{sec:other-apps}, we give the details of the applications of shortcut partition in constructing tree cover, distance oracle, and additive embedding into bounded treewidth graphs.

\section{Reduction to Shortcut Partition}\label{sec:SPR-from-Shortcut}

As mentioned, Filtser~\cite{Filtser20B} presented a reduction from the SPR problem to that of finding \emph{scattering partitions}. 
To prove Theorem~\ref{thm:shortcut-to-SPR}, we introduce an inherently relaxed notion of \emph{approximate} scattering partition (refer to Definition~\ref{def:scatter}),
and adapt the reduction of \cite[Theorem~1]{Filtser20B} using that notion. 
Due to the usage of our inherently relaxed notion of partition, we have to make two crucial changes in the reduction, and alter various parts of the analysis; we will point out specific changes along the way.

\begin{definition}[Approximate Scattering Partition] \label{def:scatter}
    Let $G = (V, E, \weight)$ be an edge-weighted graph. A \EMPH{$\beta$-approximate $(\tau, \Delta)$-scattering partition} of $G$ is a partition $\mathcal{C}$ of $V$ such that: 
    \begin{itemize}
        \item \textnormal{[Diameter.]} For each cluster $C$ in $\mathcal{C}$, the induced subgraph $G[C]$ has weak diameter 
        at most $\Delta$; that is, $\dist_{G}(u,v) \le \Delta$ for any vertices $u$ and $v$ in $C$.
        \item \textnormal{[Scattering.]} For any two vertices $u$ and $v$ in $V$ such that $\dist_G(u, v) \leq \Delta$, there exists a path $\pi$ in $G$ between $u$ and $v$ where (1) $\pi$ has length at most $\beta \cdot \Delta$, (2) every edge in $\pi$ has length at most $\Delta$, and (3) $\pi$ intersects at most $\tau$ clusters in $\mathcal{C}$.
        We say $\pi$ is a \EMPH{$\beta$-approximate $(\tau,\Delta)$-scattered path}.
    \end{itemize}
\end{definition}
We remark that scattering properties (2) and (3) together imply property (1):  the length of $\pi$ is at most $O(\tau)\cdot\Delta$. Nevertheless, we prefer to keep property (1) separately from properties (2) and (3) in the definition to emphasize the fact that $\pi$ is an approximate path.  
 
Notice that the notion of approximate scattering partition is more relaxed than the original notion of scattering partition \cite{Filtser20B}. 
A scattering partition requires \emph{every} shortest path with length at most $\Delta$ to be $\tau$-scattered.
However in an approximate scattering partition there are three relaxations: 
\begin{enumerate}
\item we only require that \emph{one} such path exists;
\item that path may be an approximate shortest path (rather than an exact shortest path);
\item the $\beta$-approximation to the length of such path $\pi$ is \emph{not} with respect to the distance between the endpoints;
rather, the length of $\pi$ is bounded by $\beta$ times $\Delta$, the diameter bound of clusters.  
\end{enumerate}
The following lemma, 
which we prove in \S\ref{SS:algorithm} and \S\ref{SS:distortion}, is analogous to Theorem 1 by Filtser~\cite{Filtser20B}, except for the key difference that we employ approximate scattering partitions. 
It implies that, somewhat surprisingly,
despite the three aforementioned relaxations introduced by our notion of approximate scattering partitions --- especially the third one that significantly relaxes the meaning of $\beta$-approximation --- such partitions still suffice for solving the SPR problem.

\begin{lemma}
    \label{lem:scattering-spr}
    Let $G$ be a graph such that for every $\Delta > 0$, every induced subgraph of $G$ admits a $\beta$-approximate $(\tau, \Delta)$-scattering partition, for some constants $\beta, \tau \geq 1$. Then, there is a solution to the SPR problem on $G$ with distortion $O(\tau^{8} \cdot \beta^5) = O(1)$.
\end{lemma}

To construct approximate scattering partitions, we use \emph{shortcut partitions}.
Recall the \EMPH{cluster graph} of $G$ with respect to $\mathcal{C}$, denoted \EMPH{$\check{G}$}, is the graph obtained by contracting each cluster in $\mathcal{C}$ into a \emph{supernode}. 
The \EMPH{hop-length} of a path is the number of edges in the path.

\begin{lemma} 
\label{lem:scatter}
Let $G$ be a graph and let $\Delta > 0$ be a parameter. 
If any subgraph of $G$ has an $(\eps, h)$-shortcut partition for any $\eps \in (0,1)$ and some number $h$, then $G$ has a $2\e h$-approximate $(\e h, \Delta)$-scattering~partition.
\end{lemma}

\begin{proof}
    Construct graph $G'$ from $G$ by removing all edges of length greater than $\Delta$. 
    Notice that if any pair $u,v$ of vertices satisfies $\dist_G(u,v) \le \Delta$, then it also satisfies $\dist_{G'}(u,v) \le \Delta$. Thus, any partition of vertices that satisfies the approximate scattering property for $G'$ also satisfies that property for $G$.

    Let $\mathcal{C}$ be an $(\eps, h)$-shortcut partition for the graph $G'$, 
    with parameter $\e \coloneqq \Delta/\diam(G')$.  
    Notice that $\mathcal{C}$ is a clustering of the vertices of $G$, where for any cluster $C$ in $\mathcal{C}$, the induced subgraphs $G[C]$ and $G'[C]$ have strong diameter at most $\eps \cdot \diam(G') = \Delta$; thus, $\mathcal{C}$ satisfies the diameter property of approximate scattering partition. 
    
    We now show that $\mathcal{C}$ satisfies the scattering property. Let $u$ and $v$ be two vertices in $G$ with $\dist_G(u,v) \le \Delta$.
    Note that $\dist_{G'}(u,v) \le \Delta$.
    By the properties of shortcut partition, there is a path $\check{\pi}$ in the cluster graph $\check{G}$ between the clusters containing $u$ and $v$, such that $\check{\pi}$ has hop-length at most $\e h \cdot \Ceil{ \frac{\delta_{G'}(u,v)}{\e \diam(G')} }$.
    In other words, the hop-length of $\check{\pi}$ is
    $t$, for some $t$ that is upper-bounded by $\e h$. 
    Write $\check{\pi} = (C_1, C_2, \ldots, C_t)$ as a sequence of $t$ adjacent clusters in $\check{G}$. 
    Notice that two clusters $C$ and $C'$ in $\check{G}$ are \EMPH{adjacent} if and only if there is an edge in $G'$ between a vertex in $C$ and a vertex in $C'$.
    For every pair of consecutive clusters $C_i$ and $C_{i+1}$ in $\check{\pi}$, let $x_{i}'$ be a vertex in $C_i$ and $x_{i+1}$ be a vertex in $C_{i+1}$ such that there is an edge $e_i$ in $G'$ between $x_{i}'$ and $x_{i+1}$. 
    To simplify notation, define $x_1 \coloneqq u$ and define $x_t' \coloneqq v$. 
    With this definition, $x_i$ and $x_i'$ are defined for all $i$ in $\set{1, \ldots, t}$.
    Notice that for every $i$ in $\set{1, \ldots, t}$, vertices $x_i$ and $x_i'$ are both in cluster $C_i$. 
    By the strong diameter property of $\mathcal{C}$, there is a path $P_i$ in $G'$ between $x_i$ and $x_i'$, such that $P_i$ is contained in $C_i$ and has length at most $\Delta$.

    We define the path $\pi$ in $G'$ (and thus also in $G$) between $u$ and $v$ to be the concatenation $P_1 \circ e_1 \circ P_2 \circ e_2 \circ \ldots \circ P_t$. 
    Notice that 
    (1) $\pi$ has length at most $2 t \cdot \Delta \leq 2 \e h \cdot \Delta$; indeed, each subpath $P_i$ has length at most $\Delta$ (by the strong diameter property), and each edge $e_i$ has length at most $\Delta$ (as $e_i$ is in $G'$). 
    Further, (2) every edge of $\pi$ has length at most $\Delta$, and 
    (3) $\pi$ intersects at most $\e h$ clusters (namely, the clusters $C_1, \ldots, C_t$ along $\check{\pi}$).
\end{proof}

\Cref{thm:shortcut-to-SPR} follows from \Cref{lem:scattering-spr} and \Cref{lem:scatter}. 
As a direct corollary of \Cref{thm:shortcut-minor} and \Cref{lem:scatter}, we obtain the following.

\begin{corollary}
\label{cor:scatter}
There are constants $\beta$ and $\tau$ such that, for any $K_r$-minor-free graph $G$ and any $\Delta > 0$, there exists a $\beta$-approximate $(\tau, \Delta)$-scattering partition of $G$.  Specifically, $\beta = \tau = 2^{O(r \log r)}$.
\end{corollary}
In what follows we prove Lemma~\ref{lem:scattering-spr}.

\subsection{Algorithm}
\label{SS:algorithm}

Our construction for proving Lemma~\ref{lem:scattering-spr} is similar to that of 
\cite{Filtser20B}, but deviates from it in several crucial points (see Remark~\ref{rem:diff} for details).
For completeness, we next provide the entire construction of
\cite{Filtser20B}, adapted appropriately to our purposes.

We will assume without loss of generality that the minimum pairwise distance is 1.
We shall partition $V$ into $|K|$ connected subgraphs,  
each of which corresponds to a single terminal in $K$. 
Each vertex in $V$ will be \emph{assigned} to a connected subgraph by the \EMPH{assignment function $f: V \rightarrow K$}, 
such that at the end of the process, 
we can create a graph minor $M$ of $G$ by contracting each connected subgraph $f^{-1}(t)$ into a supernode for every terminal $t \in K$. 
By setting $\weight_M(t, t') \coloneqq \dist_G(t, t')$ for each edge $(t, t') \in E(M)$, 
the edge-weighted graph $M = (K, E(M), \weight_M)$ is our solution to the SPR problem on $G$. For a path $P$, we denote by $||P||$ the length of $P$.

We compute the assignment function $f$ in iterations.
In iteration $i$ we shall compute a function $\EMPH{$f_i$}: V \rightarrow K \cup \{\perp\}$, where $\perp$ symbolizes that the vertex remains unassigned.
The function $f$ will be obtained as the function $f_i$ computed at the last iteration of the algorithm.
We will maintain the set of \EMPH{relevant vertices} $\EMPH{$\mathcal{R}_i$} \coloneqq \Set{ v \in V \mid \zeta^{i - 1} \leq \dist_G(v, K) < \zeta^i}$ and the set of \EMPH{assigned vertices $V_i$} by the function $f_i$ to some terminals, for each iteration $i$, where $\EMPH{$\zeta$} \coloneqq c \cdot \beta \cdot \tau$, for $\beta$ and $\tau$ being the constants provided by \Cref{cor:scatter} and $c$ being some large constant. 
Initialize $f_0(t) \coloneqq t$ for each $t \in K$, and $f_0(v) \coloneqq \,{\perp}$ for each $v \in V \setminus K$.
Define both \EMPH{$\mathcal{R}_0$} and \EMPH{$V_0$} to be $K$.
Inductively, we maintain the properties that $V_{i-1} \subseteq V_{i}$ and $\bigcup_{j \le i}
\mathcal{R}_j \subseteq V_i$,
hence the algorithm terminates when all vertices have been assigned.

\begin{figure}[h!]
    \centering
    \includegraphics[width=\textwidth]{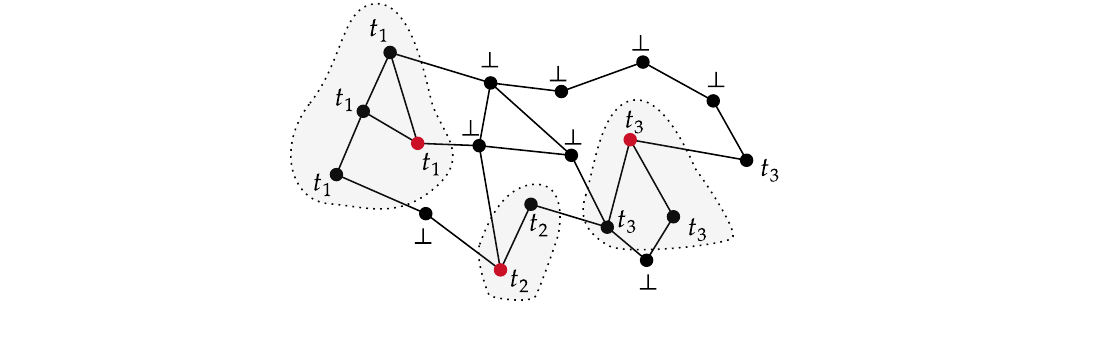}
    \caption{An SPR instance with 3 Steiner points. Values of assignment function $f$ are shown next to vertices.}
    \label{fig:SPR_f}
\end{figure}

At the $i$-th iteration of the algorithm, we compute $\beta$-approximate $(\tau, \zeta^{i-1})$-scattering partition~\EMPH{$\mathcal{P}_i$}, 
provided by \Cref{cor:scatter}, on the subgraph induced on the unassigned vertices $\EMPH{$G_i$} \coloneqq G[V \setminus V_{i-1}]$.
Let~\EMPH{$\mathcal{C}_i$} be the set of clusters in $\mathcal{P}_i$ that contain at least one vertex in $\mathcal{R}_i$. 
All vertices in the clusters of $\mathcal{C}_i$ will be assigned by $f_i$ at iteration $i$. 

We classify the clusters in $\mathcal{C}_i$
into \EMPH{levels}, starting from level 0, viewing $V_{i-1}$ as a \emph{level-0} cluster. 
We say that a cluster $C \in \mathcal{C}_i$ is at \EMPH{level $j$} if $j$ is the minimum index such that there is an edge of weight at most $\zeta^i$  connecting a vertex $u$ in $C$ and another vertex $v$ in some level-$(j-1)$ cluster $C'$.
If there are multiple such edges, we fix one of them arbitrarily;
we call vertex $v$ in $C'$ the \EMPH{linking vertex} of $C$.  
Let \EMPH{$\level_i(C)$} denote the level of $C$. 
Observe that every $\mathcal{C}_i$ contains a vertex in $\mathcal{R}_i$, i.e., there exists a vertex $v \in \mathcal{C}_i$ such that $\zeta^{i - 1} \leq \dist_G(v, K) < \zeta^i$. Hence, it is readily verified that every cluster $\mathcal{C}_i$ has a linking vertex,
and thus $\EMPH{$\level_i(C)$}$ is a valid level.

For every vertex $v \in V_{i-1}$, we set $f_{i}(v) \coloneqq f_{i-1}(v)$.
For every vertex not in $\bigcup \mathcal{C}_i$ (or $V_{i-1}$), we set $f_{i}(v) = \perp$.
Next, we scan all clusters in $\mathcal{C}_i$ by non-decreasing order of level, starting from level 1. 
For each vertex $u$ in each cluster $C$, 
we set $f_i(u)$ to be $f_i(v_C)$, where $v_C$ is the linking vertex of $C$.
If some unassigned vertices remain, we proceed to the next iteration; otherwise, the algorithm terminates.

\begin{remark} 
\label{rem:diff}
The algorithm presented here is different than that of \cite{Filtser20B} in two crucial points: 
\begin{itemize}
    \item First, as mentioned, we use approximate scattering partitions (as in Definition~\ref{def:scatter}) rather than the scattering partitions of \cite{Filtser20B}. This change poses several  technical challenges in the argument.
    \item To cope with approximate scattering partitions, we do not use constant 2 as in \cite{Filtser20B}
    but rather use a bigger constant $\zeta$ (as defined above).
\end{itemize}
\end{remark}

\subsection{Distortion Analysis}
\label{SS:distortion}

From the algorithm, any vertex within distance between $\zeta^{i-1}$ to $\zeta^i$ from $K$ is assigned at iteration at most~$i$.
However, the following claim narrows the possibilities down to two choices. The claim is analogous to Claim 5 in \cite{Filtser20B}, where we use $\zeta$ instead of 2, and its proof is similar.

\begin{claim}[{\cite[Claim~5]{Filtser20B}}]
    \label{clm:assign-itr}
    Any vertex $v$ satisfying $\zeta^{i - 1} \le \dist_G(v, K) < \zeta^i$ is assigned during iteration $i - 1$ or $i$.  
    Consequently, any vertex $v$ assigned during iteration $i$ must satisfy $\zeta^{i - 1} \le \dist_G(v, K) < \zeta^{i+1}$.
\end{claim}
\begin{proof}
If $v$ remains unassigned until iteration $i$, it will be assigned during iteration $i$ by construction.
Suppose that $v$ was assigned during iteration $j$. Then $v$ belongs to a cluster $C \in \mathcal{C}_j$, and there
is a vertex $u \in C$ with $\dist_G(u,K) \le \zeta^j$. As $C$ has strong diameter at most $\zeta^{j-1}$ and $\zeta > 2$, we obtain 
\[
\zeta^{i-1} \le \dist_G(v,K) \le
\dist_G(u,K) + \dist_G(u,v) \le \zeta^j + \zeta^{j-1} < \zeta^{2} \cdot \zeta^{j-1},
\]
implying that $i-1 < 2 + (j-1)$, or equivalently
$j \ge i-1$.
\end{proof}

\noindent The following claim is analogous to Corollary 1 in \cite{Filtser20B}, but we introduce a few changes in the proof.

\begin{claim}[{\cite[Corollary~1]{Filtser20B}}]
    \label{clm:real-dist}
    For every $v \in V$, $\dist_G(v, f(v)) \leq 3\tau \cdot \zeta^2 \cdot \dist_G(v, K)$.
\end{claim}

\begin{proof}
Let $i$ be the iteration in which $v$ is assigned, and let \EMPH{$C_v$} be the cluster in $\mathcal{C}_i$ containing $v$.     
We shall prove that 
\begin{equation}
\begin{aligned}
\label{eq:forind}
\dist_G(v, f(v)) \leq 3 \tau \cdot \zeta^{i+1}.
\end{aligned}
\end{equation}
Combining this bound with Claim~\ref{clm:assign-itr} yields 
\[
\dist(v, f(v)) \leq 3 \tau \cdot \zeta^{i+1} \leq 3 \tau \cdot \zeta^2 \cdot \dist_G(v, K),
\] 
as required.
The proof is by induction on the iteration $i$ in which $v$ is assigned. 
The base case  $i = 0$ is trivial, as then $v$ is a terminal, and we have $\dist_G(v, f(v)) = 0 \leq 3\tau \cdot \zeta^{0+1}$. 
We henceforth consider the induction step when $i \ge 1$.  

First, we argue that
$\level_i(C_v) \leq \zeta \cdot \tau$.
Since cluster $C_v$ is in $\mathcal{C}_i$, there exists a vertex $u \in C_v$ such that $\dist_G(u, K) < \zeta^i$. 
Let $P_u \coloneqq (u_1, u_2, \ldots u_s)$ be a shortest path from $u = u_1$ to $K$ (with $\len{P_u} < \zeta^i$), let $\ell$ be the largest index such that $u_1, u_2, \ldots u_{\ell} \in V \setminus V_{i - 1}$, and define the prefix $\EMPH{$Q$} \coloneqq (u_1, u_2, \ldots u_\ell)$ of $P_u$; note that $\ell < s$ and $u_{\ell+1} \in V_{i-1}$. 
Since $\len{Q} < \len{P_u} < \zeta^i$, 
we can greedily partition $Q$ into $\zeta' \le \zeta$ sub-paths $\EMPH{$Q_1, \ldots, Q_{\zeta'}$}$, each of length at most $\zeta^{i-1}$, connected via edges of weight less than $\zeta^i$; 
that is, $Q$
is obtained as the concatenation $Q_1 \circ e_1 \circ Q_2 \circ e_2 \ldots \circ e_{\zeta' -1} \circ Q_{\zeta'}$,
where $\len{Q_j} < \zeta^{i-1}$
and $\len{e_j} < \zeta^i$ for each $j$.
Consider the $\beta$-approximate $(\tau, \zeta^{i-1})$-scattering partition $\mathcal{P}_i$ (provided by \Cref{cor:scatter}), used in the $i$th iteration, on the subgraph $G_i = G[V \setminus V_{i-1}]$ induced on the unassigned vertices.
For each $j$, 
the sub-path $Q_j$ of $Q$ is contained in $G_i$
and it satisfies $\len{Q_j} \leq \zeta^{i-1}$, thus there exists a $\beta$-approximate path \EMPH{$Q'_j$} between the endpoints of $Q_j$ that is scattered by $\tau'$ clusters, with $\tau' \leq \tau$, and each edge of $Q'_j$ is of weight at most $\zeta^{i-1}$.
The path $Q'_1 \circ e_1 \circ Q'_2 \circ e_2 \ldots \circ e_{\zeta' -1} \circ Q'_{\zeta'}$
obtained from $Q$ by replacing each sub-path $Q_j$ by its scattered path $Q'_j$, is a path from $u_1$ to $u_\ell$ intersecting at most  $\zeta \cdot \tau$ clusters in $\mathcal{C}_i$. 
Since $u_{\ell}$ is in a cluster of level $1$ (because $u_{\ell+1}$ is in $V_{i-1}$, which is of level $0$), $\level_i(C_v) \leq \zeta \cdot \tau$, as required.

\medskip
\noindent We then show that $\dist_G(v, f(v)) \leq \level_i(C_v) \cdot 2 \cdot \zeta^{i} + 3\tau \cdot \zeta^{i}$ by induction on the ($i$th-iteration) level of $C_v$.
We employ a double induction, one on the iteration $i$ and the other on the level of $C_v$; aiming to avoid confusion, we shall refer to the former as the ``outer induction'' and to the latter as the ``inner induction''.

Let \EMPH{$x$} be the linking vertex of $C_v$; in particular, we have $f(v) = f(x)$. 
Let \EMPH{$x_v$} be the vertex in $C_v$ such that $(x, x_v) \in E$ and $\weight(x, x_v) \leq \zeta^i$. 
For the basis $\level_i(C_v) = 1$ of the inner induction, $x$ is assigned during iteration $i' < i$.
By the outer induction hypothesis for iteration $i'$ (i.e., substituting $i$ with $i'$ in Eq.~\ref{eq:forind}), we obtain $\delta_G(x,f(x)) \le 3 \tau \cdot \zeta^{i' + 1} \le 3 \tau \cdot \zeta^{i}$.
By the triangle inequality and
since $\zeta > 1$:
\begin{equation}
\begin{aligned}
    \dist_G(v, f(v)) &\leq \dist_G(v, x_v) + \dist_G(x_v, x) + \dist_G(x, f(v)) \\
    &\leq \zeta^{i - 1} + \zeta^{i} + \dist_G(x, f(x)) \leq \zeta^{i - 1} + \zeta^{i} + 3\tau \cdot \zeta^i \leq 2 \cdot \zeta^{i} + 3\tau \cdot \zeta^{i}.
\end{aligned}
\end{equation}
For the inner induction step, consider the case $\level_i(C_v) > 1$. Let $C_x$ be the cluster in $\mathcal{P}_i$ containing $x$; in particular, we have $\level_i(C_x) = \level_i(C_v) - 1$. By the inner induction hypothesis on the level of $C_x$, we have $\dist(x, f(x)) \leq \level_i(C_x) \cdot 2 \cdot \zeta^{i} + 3\tau \cdot \zeta^{i}$.  
Using the triangle inequality again, we have:
\begin{equation}
\begin{aligned}
    \dist_G(v, f(v)) &\leq \dist_G(v, x_v) + \dist_G(x_v, x) + \dist_G(x, f(v)) \\
    &\leq \zeta^{i - 1} + \zeta^{i} + \dist_G(x, f(x)) \leq \zeta^{i - 1} + \zeta^{i} + \level_i(C_x) \cdot 2 \cdot \zeta^{i} + 3\tau \cdot \zeta^i \\
    &\leq 2 \cdot \zeta^{i} + (\level_i(C_v) - 1) \cdot 2 \cdot \zeta^{i} + 3\tau \cdot \zeta^{i} = \level_i(C_v) \cdot 2 \cdot \zeta^{i} + 3\tau \cdot \zeta^{i},
\end{aligned}
\end{equation}
which completes the inner induction step.

Since $\level_i(C_v) \leq \zeta \cdot \tau$ and as $\zeta > 3$, it follows that $\dist(v, f(v)) \leq 3 \tau \cdot \zeta^{i+1}$, which completes the outer induction step. The claim follows. 
\end{proof}

\noindent Now we are ready to prove Lemma~\ref{lem:scattering-spr}.

\begin{proof}[of Lemma~\ref{lem:scattering-spr}]
We prove that our algorithm returns a minor of $G$ that satisfies the SPR conditions.
By the description of the algorithm, it is immediate that the subgraph induced by the vertex set $f^{-1}(t)$ is connected,  for each $t \in K$. 
Thus, it remains to prove that the minor $M$ induced by $f$ is a distance preserving minor of $G$ with distortion $O(\tau^8 \cdot \beta^5)$.

Consider an arbitrary pair of terminals $t$ and $t'$. Let $P \coloneqq (v_1, v_2, \ldots, v_{|P|})$ be a shortest path between $v_1 \coloneqq t$ and $v_{|P|} \coloneqq t'$.
For each subpath $I \coloneqq (v_\ell, v_{\ell + 1}, \ldots v_r)$ of $P$, let \EMPH{$I^+$} denote the \EMPH{extended subpath} $(v_{\ell - 1}, v_\ell, v_{\ell + 1}, \ldots v_r, v_{r + 1})$; we define $v_0 \coloneqq v_1$ and $v_{|P| + 1} \coloneqq v_{|P|}$ for technical convenience. 
Partition $P$ into a set \EMPH{$\mathcal{I}$} of subpaths called \EMPH{intervals} such that for each subpath $I \in \mathcal{I}$ between $v_\ell$ and $v_r$: 
\begin{equation}
\begin{aligned}
    \label{eq:wI}
    \len{I}  \leq \eta \cdot\dist_G(v_\ell, K) \leq \len{I^+},
\end{aligned}
\end{equation}
where $\EMPH{$\eta$} \coloneqq \frac{1}{4\zeta}$.
It is easy to verify that $\mathcal{I}$ can be constructed greedily from $P$.

Consider an arbitrary interval $I = (v_\ell, v_{\ell + 1}, \ldots v_r) \in \mathcal{I}$. 
Let $u \in I$ be a vertex that is assigned in iteration $i$, and assume no vertex of $I$ was assigned prior to iteration $i$.
Since $u$ is assigned in iteration $i$, $u$ belongs to a cluster $C$ in $\mathcal{C}_i$, which is the subset of
clusters that contain at least one vertex in $\mathcal{R}_i$, among the $\beta$-approximate $(\tau, \zeta^{i-1})$-scattering partition $\mathcal{P}_i$ computed at the $i$th iteration.
Hence, by definition, $C$ has strong diameter  at most $\zeta^{i-1}$ and there exists a vertex $u' \in C$ such that $\dist_G(u', K) < \zeta^i$, implying that 
\begin{equation}
\begin{aligned}
    \label{eq:duK2}
    \dist_G(u, K) \leq 
    \dist_G(u, u') +
    \dist_{G}(u', K) 
    < 
    \zeta^{i-1} + \zeta^i
    < 
    2\zeta^{i}.   
\end{aligned}
\end{equation}
By Eq.~\ref{eq:wI} and the triangle inequality, 
\[
\dist_G(v_\ell, K) \leq \dist_G(v_\ell, u) + \dist_G(u, K) \leq \len{I} + \dist_G(u, K)
\leq \eta \cdot \dist_G(v_\ell,K) + \dist_G(u,K),
\]
which together with Eq.~\ref{eq:duK2} and the fact that $\eta < 1/2$ yields
\begin{equation}
\begin{aligned}
\label{eq:boundvlk}
    \dist_G(v_\ell,K) \le \frac{\dist_G(u,K)}{1-\eta} <  \frac{2\zeta^{i}}{1-\eta} 
    <  4 \zeta^{i}. 
\end{aligned}
\end{equation}
By Eq.~\ref{eq:wI} and Eq.~\ref{eq:boundvlk}, 
\begin{equation}
\begin{aligned}
\label{eq:distsmall}
    \dist_G(v_\ell, v_r) &= \len{I} \leq \eta \cdot \dist_G(v_\ell, K) < \eta \cdot 4\zeta^{i}  = \zeta^{i - 1},
\end{aligned}
\end{equation}
where the last inequality holds
as $\eta = \frac{1}{4\zeta}$.
    
At the beginning of iteration $i$, all vertices of $I$ are unassigned, i.e., $I$ is in $G_i = G[V \setminus V_{i-1}]$,
    and Eq.~\ref{eq:distsmall} yields
$\dist_{G_{i}}(v_\ell, v_r)
= \dist_{G}(v_\ell, v_r)
<  \zeta^{i-1}$.
At the $i$th iteration a $\beta$-approximate $(\tau, \zeta^{i - 1})$-scattering partition 
$\mathcal{P}_{i}$
on $G_{i}$ is computed, thus there exists a $\beta$-approximate 
$(\tau, \zeta^{i - 1})$-scattered
path \EMPH{$I'$} in $G_{i}$ from $v_\ell$ to $v_r$ that is scattered by at most $\tau$ clusters in $\mathcal{P}_{i}$, with $\len{I'} \le \beta \cdot \zeta^{i-1}$.
A path is called a \EMPH{detour} if its first and last vertices are assigned to the same terminal. 
Since vertices in the same cluster will be assigned to the same terminal,
at the end of iteration $i$, $I'$ can be greedily partitioned into at most $\tau$ detours and $\tau + 1$ subpaths that contain only unassigned vertices;
in other words, we can write $I' \coloneqq \EMPH{$P_1 \circ Q_1 \circ \ldots \circ P_\rho \circ Q_\rho \circ P_{\rho + 1}$}$, where $\rho \leq \tau$, $Q_1, Q_2, \ldots Q_{\rho}$ are detours, and each of the (possibly empty) sub-paths $P_1, P_2, \ldots P_{\rho + 1}$ contains only unassigned vertices at the end of iteration $i$. 

Fix an arbitrary index $j \in [1 \,..\,\rho+1]$.
Let $a_j$ and $b_j$ be the first and last vertices of $P_j$; it is possible that $a_j = b_j$.  
Since $\len{I'} \le \beta \cdot \zeta^{i-1}$ and as $\beta < \zeta$, we have
\begin{equation}
\begin{aligned} \label{eq:ajbj}
    \dist_G(a_j, b_j) \leq \len{P_j} \le \len{I'} \leq \beta \cdot \zeta^{i-1} < \zeta^i.    
\end{aligned}
\end{equation}
At the beginning of iteration $i+1$, all vertices of $P_j$ are unassigned by definition, hence $P_j$ is in $G_{i+1} = G[V \setminus V_{i}]$
and by Eq.~\ref{eq:ajbj},
$\dist_{G_{i+1}}(a_j, b_j)
\le \len{P_j} < \zeta^i$.
At the $(i+1)$th iteration a
$\beta$-approximate $(\tau, \zeta^{i})$-scattering partition  $\mathcal{P}_{i+1}$ on $G_{i+1}$ is computed, thus there exists a $\beta$-approximate 
$(\tau, \zeta^{i})$-scattered path \EMPH{$P'_j$} 
 in $G_{i+1}$ from $a_j$ to $b_j$
that is scattered by at most $\tau$
clusters in $\mathcal{P}_{i+1}$, with $\len{P'_j} \leq \beta \cdot \zeta^i$. 
    
Next, consider the path $\EMPH{$I''$} \coloneqq P'_1 \circ Q_1 \circ \ldots \circ P'_\rho \circ Q_\rho \circ P'_{\rho + 1}$. 
By Eq.~\ref{eq:distsmall} we have 
\begin{equation}
\begin{aligned} 
\label{eq:boundi''}
\len{I''} \leq 
\len{I} + \sum_{j=1}^{\rho+1} \len{P'_j} 
\le \zeta^{i-1} + (\tau+1) \beta \cdot \zeta^i 
\le (\tau+2) \beta \cdot \zeta^i
\end{aligned}
\end{equation}
Since no vertex in $I$ (in particular, $v_\ell$) was assigned prior to iteration $i$, 
Claim~\ref{clm:assign-itr} yields
$\dist_G(v_\ell,K) \ge \zeta^{i-1}$.
Eq.~\ref{eq:wI} yields $\len{I^+} \ge \eta \cdot \dist_G(v_\ell,K) \ge \eta \cdot \zeta^{i-1}$,
and as $\eta = \frac{1}{4\zeta}$ 
we obtain
\begin{equation}
\begin{aligned} 
\label{eq:idouble}
\len{I''} \leq  (\tau+2) \beta \cdot \zeta^i
\le 4 \zeta^2 (\tau+2)\beta 
\cdot \len{I^+}. 
\end{aligned}
\end{equation}
    
Next, we argue that
all vertices in $I''$ are assigned at the end of iteration $i + 1$.
Let \EMPH{$w$} be an arbitrary vertex in $I''$;
by Claim~\ref{clm:assign-itr}, it suffices to show that $\dist_G(w, K) < \zeta^{i + 1}$. 
Recall that $u$ is a vertex of $I$ that is assigned in iteration $i$.
By Eq.~\ref{eq:duK2}, Eq.~\ref{eq:distsmall}, Eq.~\ref{eq:boundi''} and the triangle inequality, 
\begin{equation}
\begin{aligned} 
\label{eq:dw}
    \dist_G(w, K) &\leq \dist_G(v_\ell, K) + \dist_G(v_\ell, w) \leq \dist_G(v_\ell, u) + \dist_G(u, K) + \dist_G(v_\ell, w) \\
    &\leq \len{I} + \dist_G(u, K) + \len{I''} < \zeta^{i-1} + 2\zeta^{i} + (\tau+2) \beta \cdot \zeta^i < \zeta^{i+1}, 
\end{aligned}
\end{equation}
where the last inequality holds since
$\zeta = c \cdot \beta \cdot \tau$ for a sufficiently large constant $c$.

Hence, every vertex in $P'_j$ is assigned by iteration $i + 1$, for every $j \in [1 \,..\,\rho+1]$. Then, $P'_j$ could be greedily partitioned into at most $\tau$ detours, as before with $I'$, but we have no subpaths of unassigned vertices in $I''$, since every vertex in $I''$ must be assigned by the end of iteration $i + 1$. 
We have thus shown that $I''$ can be partitioned into at most $O(\tau^2)$ detours \EMPH{$D_1, D_2, \ldots D_g$}, with $\EMPH{$g$} = O(\tau^2)$. 
For each $j \in [1 \,..\, g]$, let \EMPH{$x_j$} and \EMPH{$y_j$} be the first and last vertices in $D_j$.
Because $I''$ are partitioned greedily into \emph{maximal} detours, one has $f(y_j) \ne f(x_{j+1})$ for all $j$.
Observe that there exists an edge between $f(x_j)$ and $f(x_{j + 1})$ in the SPR minor $M$ for each $j \in [1 \,..\, g-1]$, since $f(x_j) = f(y_j) \in K$ and $(y_j, x_{j + 1}) \in E$.
Consequently, by the triangle inequality, Corollary~\ref{clm:real-dist} and Eq.~\ref{eq:idouble},
\begin{equation}
\begin{aligned} 
\label{eq:dvlvr1}
\dist_M(f(v_\ell), f(v_r)) 
&\leq \sum_{j = 1}^{g - 1}\dist_M(f(x_j), f(x_{j + 1})) = \sum_{j = 1}^{g - 1}\dist_G(f(x_j), f(x_{j + 1})) \\
&\leq \sum_{j = 1}^{g - 1}\left[\dist_G(x_j, f(x_j)) + \dist_G(x_j, x_{j + 1}) + \dist_G(x_{j + 1}, f(x_{j + 1}))\right] \\
&\leq 2\sum_{j = 1}^{g} \dist_G(x_j, f(x_j)) + \sum_{j = 1}^{g - 1}\dist_G(x_j, x_{j + 1}) \leq 2\sum_{j = 1}^{g} \dist_G(x_j, f(x_j)) + \len{I''} \\
&\leq 6\tau \zeta^2 \sum_{j = 1}^{g} \dist_G(x_j, K) + 4 \zeta^2 (\tau+2)\beta 
\cdot \len{I^+}. 
\end{aligned}
\end{equation}
For every vertex $v'' \in I''$, we have
\begin{equation}
\begin{aligned} 
\label{eq:dv-k}
    \dist_G(v'', K) &\leq \dist_G(v'', v_\ell) + \dist_G(v_\ell, K) \leq \len{I''} + \dist_G(v_\ell, K) \\
    &\leq 4 \zeta^2 (\tau+2)\beta 
    \cdot \len{I^+}  + \frac{\len{I^+}}{\eta} \leq 4 \zeta^2 (\tau+3)\beta 
    \cdot \len{I^+},
\end{aligned}
\end{equation}
where the penultimate inequality holds by
Eq.~\ref{eq:wI} and Eq.~\ref{eq:idouble} and the last inequality holds
since $\eta = \frac{1}{4\zeta}$. We remark that  Eq.~\ref{eq:dv-k} also holds for any vertex $v' \in I$, which will be used below for deriving Eq.~\ref{eq:dist-inter-interval}. Hence, for every $j \in [1 \,..\, g]$, $\dist_G(x_j, K) \leq 4 \zeta^2 (\tau+3)\beta 
\cdot \len{I^+}$; plugging this in Eq.~\ref{eq:dvlvr1} yields:
\begin{equation}
\begin{aligned}
\label{eq:dvlvr2}
    \dist_M(f(v_\ell), f(v_r)) \leq 24
    \zeta^4 \tau(\tau+3) \beta g \cdot \len{I^+}
    + 4 \zeta^2 (\tau+2)\beta 
\cdot \len{I^+}
     = O(\zeta^4 \cdot \tau^4 \cdot \beta)
     \cdot \len{I^+}.
\end{aligned}
\end{equation}

Next, we bound the distance between $t$ and $t'$ in $M$. 
So far we fixed an arbitrary interval  $I = (v_\ell, v_{\ell + 1}, \ldots v_r) \in \mathcal{I}$.
Writing $\mathcal{I} = \{I_1, I_2, \ldots I_s\}$, we have 
$\sum_{j = 1}^s\len{I_j} = \len{P} = \dist_G(t,t')$, hence  
    \begin{equation}
\begin{aligned}
    \label{eq:basicsum}
        \sum_{j = 1}^s\len{I^+_j} \leq 2\len{P} = 2\cdot \dist_G(t, t'). 
    \end{aligned}
\end{equation}
For each $I_j$, let \EMPH{$v^j_\ell$} and \EMPH{$v^j_r$} be the first and last vertices of $I_j$. 
For each $j \in [1 \,..\, s-1]$, since $(v_r^j, v_\ell^{j + 1}) \in E$, either $(f(v_r^j), f(v_\ell^{j + 1})) \in E(M)$ or
$f(v_r^j) = f(v_\ell^{j + 1})$,
thus we have 
$\dist_M(f(v^j_r), f(v^{j + 1}_\ell))
= \dist_G(f(v^j_r), f(v^{j + 1}_\ell))$. Hence, using the triangle inequality:
\begin{equation}
\begin{aligned}
\label{eq:dtt'}
    \dist_M(t, t') &\leq \sum_{j = 1}^{s - 1}
    \Paren{\big. \dist_M(f(v^j_\ell), f(v^j_r)) + \dist_M(f(v^j_r), f(v^{j + 1}_\ell)) } 
    + \dist_M(f(v^s_\ell), f(v^s_r)) \\
    &\leq O(\zeta^4 \cdot \tau^4 \cdot \beta) \cdot \sum_{j = 1}^s \len{I_j^+} + \sum_{j = 1}^{s - 1}\dist_M(f(v^j_r), f(v^{j + 1}_\ell)) \qquad \text{(by Eq.~\ref{eq:dvlvr2})} \\
    &\leq O(\zeta^4 \cdot \tau^4 \cdot \beta) \cdot \dist_G(t, t') + \sum_{j = 1}^{s - 1}\dist_M(f(v^j_r), f(v^{j + 1}_\ell)).
    \qquad \text{(by Eq.~\ref{eq:basicsum})}
    \\
    &= O(\zeta^4 \cdot \tau^4 \cdot \beta) \cdot \dist_G(t, t') + \sum_{j = 1}^{s - 1}\dist_G(f(v^j_r), f(v^{j + 1}_\ell)).
\end{aligned}
\end{equation}
Using the triangle inequality again, we have: 
\begin{equation}
\begin{aligned}
\label{eq:dist-inter-interval}
    \sum_{j = 1}^{s - 1}\dist_G(f(v^j_r), f(v^{j + 1}_\ell)) &\leq \sum_{j = 1}^{s - 1}
    \Paren{\big. 
    \dist_G(f(v^j_r), v^j_r) + \dist_G(v^j_r, v^{j + 1}_\ell) + \dist_G(v^{j + 1}_\ell, f(v^{j + 1}_\ell)) } \\
    &\leq \sum_{j = 1}^{s - 1}\dist_G(v^j_r, v^{j + 1}_\ell) + \sum_{j = 1}^{s} \Paren{\big.
    \dist_G(v^j_\ell,f(v^j_\ell)) + \dist_G(v^j_r,f(v^j_r)) } \\
    &\leq \len{P} + 3\tau \zeta^2 \cdot \sum_{j = 1}^{s}
    \Paren{\big.
    \dist_G(v^j_\ell, K) + \dist_G(v^j_r, K) } \qquad \text{(by Corollary~\ref{clm:real-dist})} \\
    &\leq \dist_G(t, t') + 3\tau \zeta^2 \cdot 4 \zeta^2 (\tau+3)\beta \cdot \sum_{j = 1}^{s}(\len{I_j^+} + \len{I_j^+})  \qquad \text{(by Eq.~\ref{eq:dv-k})} \\ 
    &\leq O(\zeta^4 \cdot \tau^2 \cdot \beta) \cdot \dist_G(t, t') \qquad \text{(by Eq.~\ref{eq:basicsum})}. 
\end{aligned}
\end{equation}

\noindent Plugging Eq.~\ref{eq:dist-inter-interval} into Eq.~\ref{eq:dtt'}, we obtain $\dist_M(t, t') = O(\zeta^4 \cdot \tau^4 \cdot \beta) \cdot \dist_G(t, t')$.
Since $\zeta = O(\beta \cdot \tau)$, we conclude that 
$\dist_M(t, t') = O(\tau^8 \cdot \beta^5)$, as required.
\end{proof}


\section{Buffered cop decompositions for minor-free graphs}
\label{SS:buffered-cop}

\begin{table}[h!]\small
\centering
\smallskip
\def\arraystretch{1.3}
\begin{tabular}{c|l}
\emph{notation} & \emph{meaning} \\ \hline 
$\cT$ & \emph{partition tree}: nodes of $\cT$ are supernodes
\\ \hdashline
$\mathcal{\hat{T}}$ & \emph{expansion of $\cT$} 
\\ \hdashline
$\eta$ & \makecell[l]{%
\emph{supernode}: \emph{induced subgraph} on its vertices (one may identify $\eta$ with these vertices); \\
$\eta$ contains $T_\eta$ initially and may only grow}
\\ \hdashline
$\dom(\eta)$ & \emph{domain of $\eta$}: subgraph induced by the union of all supernodes in the subtree of $\cT$ rooted at $\eta$
\\ \hdashline
$T_\eta$ & \emph{tree skeleton}: SSSP tree in $\dom(\eta)$ (remain fixed) 
\\ \hdashline
$\cS$ & \emph{set of supernodes}: $\cS$ starts empty and grows, 
and each supernode may grow
\\ \hdashline
witness $v_S$ & vertex adjacent to some vertex in supernode $S$
\\ \hdashline
$H$ sees $S$ & subgraph $H$ has a witness vertex $v_S$ to $S$
\\ \hdashline
$\seen{H}$ & set of supernodes in $\cS$ subgraph $H$ can see
\\ \hdashline
$\dom_{\cS}(\eta)$ & \emph{domain of $\eta$ with respect to $\cS$}: may shrink; the final $\dom_{\cS}(\eta)$ is $\dom(\eta)$
\\ \hdashline
$\bdry H'_{\downarrow X}$ & \makecell[l]{%
\emph{boundary vertices}: vertices in $G \setminus H'$ that are (1) adjacent to $H'$, and (2) in  $\dom_{\cS}(X)$}
\\ \hdashline
$\buff H'_X$ & \makecell[l]{%
\emph{buffer vertices}: unassigned vertices in $H'$ within distance (in $\dom_{\cS}(X)$) $\Delta/r$ of $\bdry H'_{\downarrow X}$}
\\ \hdashline
$\eta_{\cS}$ & vertices assigned to $\eta$ by the current $\cS$
\end{tabular}
\caption{Glossary for the construction of buffered cop decompositions.}
\label{T:glossary}
\end{table}
\paragraph{Buffered cop decomposition.}

Let $G$ be a graph.
A \EMPH{supernode $\eta$} with \EMPH{skeleton $T_\eta$} and \EMPH{radius~$\Delta$} is an induced subgraph $\eta$ of $G$ containing a tree $T_\eta$ where every vertex in $\eta$ is within distance $\Delta$ of $T_\eta$ for some real number $\Delta$, where distance is measured with respect to $\eta$.
We occasionally abuse notation and use $\eta$ to refer to the set of vertices in $\eta$, rather than the subgraph.
A \EMPH{buffered cop decomposition} for $G$ is a partition of $G$ into vertex-disjoint supernodes, together with a tree~\EMPH{$\cT$} called the \EMPH{partition tree}, whose nodes are the supernodes of $G$. 
For any supernode~$\eta$, the \EMPH{domain $\dom(\eta)$} denotes the subgraph induced by the union of all vertices in supernodes in the subtree of $\cT$ rooted at $\eta$.

\begin{definition}\label{def:buffer-cop}
    A \EMPH{$(\Delta, \gamma, w)$-buffered cop decomposition}
    for $G$ is a buffered cop decomposition $\cT$ that satisfies the following properties:
    \begin{itemize}
    
    \item \textnormal{[Supernode radius.]} 
    Every supernode $\eta$ has radius at most $\Delta$.
    
    \item \textnormal{[Shortest-path skeleton.]} 
    For every supernode $\eta$, the skeleton $T_\eta$ is an SSSP tree in $\dom(\eta)$, with at most $w$ leaves.
    
    \item \textnormal{[Supernode buffer.]} 
    Let $\eta$ be a supernode, and let $X$ be another supernode that is an ancestor of $\eta$ in the partition tree $\cT$.
    Then either $\eta$ and $X$ are adjacent in $G$, or
    for every vertex $v$ in $\dom(\eta)$, we have $\dist_{\dom(X)}(v, X) \ge \gamma$.
    
    \end{itemize}
\end{definition}

Our definition of buffered cop decomposition is different from the cop decomposition in the prior work~\cite{andreae1986pursuit,agg14}. 
Recall that the cop decomposition in prior work is a tree decomposition, where each bag of the tree decomposition can be partitioned into $r-1$ supernodes. 
Here each node of the partition tree $\mathcal{T}$ in our definition is exactly one supernode. 
We find that this alternative definition helps to simplify the presentation of our buffer-creating algorithm significantly.

Given a partition tree $\mathcal{T}$, we construct another tree \EMPH{$\mathcal{\hat{T}}$} from (and isomorphic as a graph to) $\mathcal{T}$ as follows: for each supernode $\eta\in \mathcal{T}$, we create a corresponding node \EMPH{$B_{\eta}$}, called the \EMPH{bag of $\eta$}, containing  $\eta$ and all the ancestor supernodes adjacent to $\eta$ in $G$.  Intuitively $B_{\eta}$ corresponds to the existing supernodes that $\eta$ can ``see''. Notice that $B_\eta$ is the highest bag that contains $\eta$.
By identifying each supernode with its vertex set, each bag in $\hat{\cT}$ naturally corresponds to a set of vertices in $G$.
We call $\mathcal{\hat{T}}$ the \EMPH{expansion} of $\mathcal{T}$. 
The expansion $\mathcal{\hat{T}}$ of $\mathcal{T}$ is the cop decomposition in the sense of prior work~\cite{andreae1986pursuit,agg14} discussed above, and that's why we call nodes of $\mathcal{\hat{T}}$ \emph{bags}.
While it is not immediately clear that $\mathcal{\hat{T}}$ is a tree decomposition based on its definition,
our construction guarantees that $\mathcal{\hat{T}}$ indeed satisfies all the properties of a tree decomposition. 

\begin{itemize}
    \item \textnormal{[Tree decomposition.]} There exists an expansion $\mathcal{\hat{T}}$ of $\mathcal{T}$ such that (1) $\mathcal{\hat{T}}$ is a tree decomposition of $G$, and 
    (2) every bag of $\mathcal{\hat{T}}$ contains at most $w$ supernodes. 
\end{itemize}

\noindent
We say that such a buffered cop decomposition $\cT$ has \EMPH{radius $\Delta$}, \EMPH{buffer $\gamma$}, and \EMPH{width $w$}. See \Cref{fig:buffered-cop} for an illustration, and \Cref{T:glossary} for a glossary of terminologies for the buffered cop decompositions.

\begin{figure}[h!]
    \centering
    \includegraphics[width=0.9\textwidth]{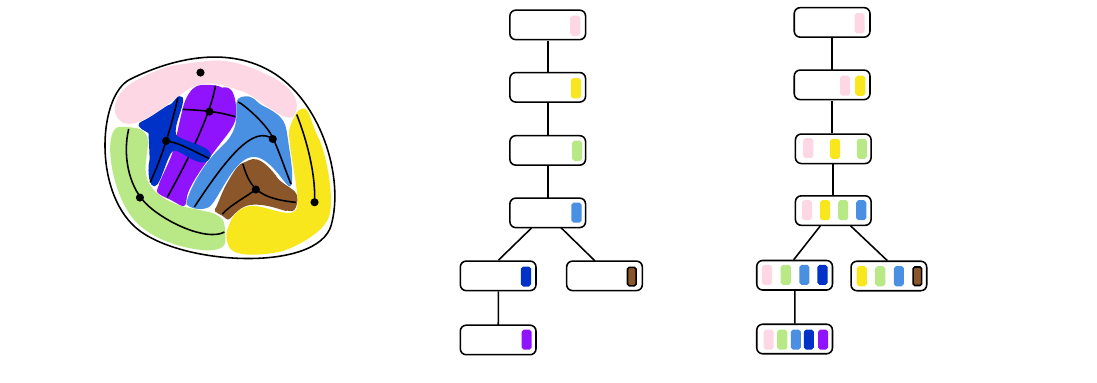}
    \caption{Left: A \emph{non-planar} graph $G$ with a partition into supernodes.  Notice that the purple cluster is connected and goes behind the dark blue supernode.
    Middle: The partition tree $\cT$ of a buffered cop decomposition for $G$. 
    The supernode buffer property guarantees that any path between the brown and pink supernodes is of length at least $\gamma$.
    Right: The expansion of $\cT$, where each bag contains at most $5$ supernodes. }
    \label{fig:buffered-cop}
\end{figure}

Given a $K_r$-minor-free graph and a parameter $\Delta$, we will construct a buffered cop decomposition with radius $\Delta$, buffer $\Delta/r$, and width $r - 1$. 
We emphasize that the most interesting property is the supernode buffer property, which says that if a supernode $X$ gets ``cut off'' from part of the graph, there is a ``buffer region'' of at least $\gamma$ between $X$ and that part of the graph. 
More precisely, let $G'$ be the subgraph of $G$ induced by vertices in descendant supernodes of $X$ that are not adjacent to $X$. 
(That is, $X$ is  ``cut off'' from $G'$ by the descendant supernodes that are adjacent to $X$.) 
The supernode buffer property in \Cref{def:buffer-cop} implies that $\dist_{\dom(X)}(v,X)\geq \gamma$ for every $v\in V(G')$. 
The construction of \cite{agg14} produces a cop decomposition with the other three properties; that is, a buffered cop decomposition with radius $\Delta$ and width $r - 1$. 
A delicate argument shows that their construction achieves something similar to a supernode buffer of $\Delta/r$ \emph{in expectation}.

\paragraph{Review of the construction of \cite{agg14}.} 
The construction of \cite{agg14} iteratively builds a collection $\cS$ of supernodes of a graph $G$. 
At each point in the algorithm, they process a subgraph $H$ of $G$ by creating a new supernode $\eta$ in $H$, and then recursing on the connected components of $H \setminus \eta$.

To describe how to create each new supernode $\eta$, we introduce some terminology. 
A subgraph $H$ \EMPH{sees} a supernode $S$ if (1) $S$ is disjoint from $H$, and (2) there exists some \EMPH{witness vertex $v_S$} in $H$ that is adjacent to a vertex in $S$. 
For any subgraph $H$, let \EMPH{$\seen{H}$} be the set of supernodes that $H$ sees.
The algorithm of \cite{agg14} guarantees that, if $G$ excludes a $K_r$-minor, the subgraph $H$ (at any point in the algorithm) sees at most $r - 2$ previously-created supernodes. Their algorithm has the following steps:

\begin{enumerate}
    \item \emph{Initialize a new supernode $\eta$.}
    
    Choose an arbitrary vertex $v$ in $H$. 
    Build a shortest-path tree $T_\eta$ in $H$ that connects $v$ to an arbitrary witness vertex for every supernode seen by $H$. 
    Initialize supernode $\eta \gets T_\eta$ with skeleton $T_\eta$.
        
    \item \emph{Expand $\eta$, to guarantee supernode buffer property in expectation.}

    Let $\gamma$ be a random number between $0$ and $C\cdot\Delta$ (for some constant $C$) drawn from a truncated exponential distribution with rate $O(r/\Delta)$, meaning that $\mathbb{E}[\gamma] = O(\Delta/r)$. 
    Assign every vertex within distance $\gamma$ of $T_\eta$ to be a part of supernode $\eta$ (where distances are with respect to $H$).

    \item \emph{Recurse.}
    
    Recurse on each connected component in the graph in $H \setminus \eta$.

\end{enumerate}

The subgraph $H$ is initially selected to be $G$. The buffered cop decomposition is implicit from the recursion tree.  
They show that at any point in the algorithm, the set of supernodes seen by $H$ forms a model of a complete graph (see \Cref{lem:seen-width}); this proves the bag width property. The tree decomposition, radius, and shortest-path skeleton properties are all straightforward to verify. 
The proof of the ``expected'' supernode buffer property is quite complicated, and requires dealing with the fact that $\gamma$ is drawn from a truncated exponential distribution rather than a normal exponential distribution. 

Here we remark that, while their buffer guarantee is in expectation, the nature of their buffer property is somewhat stronger than ours: whenever a new skeleton $T_{\eta}$ is added and cuts off the shortest path from a vertex $v$ to another skeleton $X$ (which could still be adjacent to $T_{\eta}$), the distance from $v$ to $T_{\eta}$ is smaller than the distance from $v$ to $X$ by $O(\Delta/r)$ in expectation. Here we only guarantee the distance reduction from $v$ to $X$ when $X$ and $T_{\eta}$ are not adjacent.

\subsection{Construction}

We modify the algorithm of Abraham \etal~\cite{agg14} to obtain the (deterministic) supernode buffer property. 
Throughout our algorithm, we maintain the global variables \EMPH{$\cS$}, indicating the set of supernodes, and \EMPH{$\cT$}, indicating the partition tree. At any moment
during the execution of our algorithm, some vertices of graph $G$ will already be assigned to supernodes, and some vertices will be unassigned. At the end of the execution, all vertices will be assigned by $\cS$.
At each stage of the algorithm, we (1) select some unassigned vertices to become a new supernode $\eta$, (2) assign some unassigned vertices to existing supernodes (\emph{not necessarily $\eta$}) to guarantee the supernode buffer property, and (3) recurse on connected components induced by the remaining unassigned vertices.

Our main procedure is \EMPH{$\textsc{BuildTree}(\cS, H)$}, which takes as input a connected subgraph $H$ induced by unassigned vertices in $G$. 
It assigns vertices in $H$ to supernodes in $\cS$, and returns a buffered cop decomposition. 
Figure~\ref{fig:example-algo} gives an example;
Figure~\ref{alg:buildtree} gives the complete pseudocode.
The algorithm consists of the following steps:
\begin{enumerate}

    \item \emph{Initialize a new supernode.}
    
    Choose an arbitrary vertex $v$ in $H$. Build a shortest path tree $T_\eta$ in $H$ that connects $v$ to an arbitrary witness vertex for every supernode seen by $H$. 
    Initialize supernode $\eta$ to be the subgraph of $G$ induced by all vertices of $T_\eta$; set $T_\eta$ to be the skeleton of $\eta$; and add $\eta$ to $\cS$.
    Define the domain of $\eta$ with respect to $\cS$, \EMPH{$\dom_\mathcal{S}(\eta)$}, to be the set of all vertices in $H$ that are not assigned (by $\cS$) to any supernode above $\eta$ in the partition tree $\cT$;
    initially $\dom_{\cS}(\eta) = H$, and at the end of the algorithm it will hold that $\dom_{\cS}(\eta) = \dom(\eta)$.
    (Notice that $\eta$ will grow and $\dom_{\cS}(\eta)$ will shrink over the course of the algorithm as $\cS$ changes, though $T_\eta$ will remain unchanged.  See \Cref{C:basic}(\ref{clm:grow-n-shrink}).) 

    \item \emph{Assign vertices to existing supernodes, to guarantee the supernode buffer property.}
    
    For each connected component $H'$ of $H \setminus \eta$, consider the set of supernodes
    \EMPH{$\cX$} that \emph{can} be seen by $H$ but \emph{cannot} be seen by $H'$. These supernodes are ``cut off'' from $H'$.
    In this step, we identify every \emph{currently unassigned} vertex that could be close to a cut-off supernode,
    and assign those vertices to some existing supernode (possibly to the newly-created $\eta$).

    In more detail: For each $X$ in $\cX$, define the \EMPH{boundary vertices} \EMPH{$\bdry H'_{\downarrow X}$} to be the set of vertices in $G \setminus H'$ that are (1) adjacent to $H'$, and (2) in $\dom_{\cS}(X)$. Our algorithm will maintain the invariant that all vertices adjacent to $H'$ (in particular, all vertices in $\bdry H'_{\downarrow X}$) have already been assigned to a supernode by $\cS$; see \Cref{I:adj-assignment-weak} for the formal statement.
    Define the set of \EMPH{buffer vertices $\buff H'_X$} to be the set of unassigned vertices in $H'$ within distance $\Delta/r$ of $\bdry H'_{\downarrow X}$, where distance is measured with respect to $\dom_{\cS}(X)$.
    Assign each vertex in $\buff H'_X$ to the same supernode as a closest vertex
    in $\bdry H'_{\downarrow X}$, breaking ties consistently and measuring distance with respect to $\dom_{\cS}(X)$; notice that ``the supernode of a vertex in $\bdry H'_{\downarrow X}$'' is well-defined because of \Cref{I:adj-assignment-weak}.
    
    This procedure may cut off $H'$ from another supernode, even if $H'$ may originally have been able to see that supernode (even $\eta$ itself could become cut off at this point); and it may break $H'$ into multiple connected components. Repeat this assignment process on each connected component until we have dealt with all supernodes that have been cut off.

    In \Cref{lem:buffer}, we show that this procedure guarantees that the supernode buffer property holds.
    It will suffice to show that, in this step, we assign every vertex in $H'$ that could become close to some cut-off supernode $X$, \emph{even if $X$ grows in the future}. Crucially, in this step we assign every vertex in $H'$ that is within $\Delta/r$ distance of the boundary $\bdry H'_{\downarrow X}$. It would \emph{not} suffice to just assign vertices within $\Delta/r$ distance of $X$ in the current step, because $X$ could potentially grow in the future, and the distance from a vertex in $H'$ to $X$ could shrink. We show that even if $X$ expands in the future, it remains disjoint from $X'$; further, we show that every path in $\dom(X)$ from a vertex in $H'$ to a vertex outside of $H$ passes through some boundary vertex in $\bdry H'_{\downarrow X}$.
    Thus, every vertex in $H'$ is closer\footnote{We assume the weight of every edge in $G$ is nonzero.}
    to $\bdry H'_{\downarrow X}$ than to $X$ (which are always outside $H'$), even if $X$ expands in the future. 
    This means that the vertices of $\buff H'_X$ form a buffer of $\Delta/r$ between $X$ and the unassigned vertices of $H'$. Note that we assign each vertex in $\buff H'_X$ to some supernode that is not $X$ (as $H'$ does not see $X$, no vertex in $\buff H'_{\downarrow X}$ is in $X$).
    
    This procedure is called \EMPH{$\textsc{GrowBuffer}(\cS, \cX, H')$}; see pseudocode in Figure~\ref{alg:growbuffer}.
    It takes as input a subgraph $H'$ and a list $\cX$ of supernodes that have been cut off from $H'$. 
    It assigns some vertices in $H'$ to existing supernodes in $\cS$.

    \item \emph{Recurse.}
    
    For each connected component $H'$ in the graph induced by unassigned vertices, recursively call $\textsc{BuildTree}(\cS, H')$.
    
\end{enumerate}

\noindent To initialize, let $\cS \gets \varnothing$, and call $\textsc{BuildTree}(\cS, G)$ to produce a buffered cop decomposition $\cT$ for~$G$. Throughout the algorithm, we maintain the following invariant:

\begin{invariant}
\label{I:adj-assignment-weak}
    Suppose that call $C$, whether it is $\textsc{GrowBuffer}(\cS, \cX, H)$ or $\textsc{BuildTree}(\cS, H)$, is made at some point in the algorithm. At the time call $C$ is made, every vertex in $H$ is unassigned, and every vertex in $G \setminus H$ that is adjacent to $H$ is already assigned to some supernode.
\end{invariant}
(We say that a vertex $v$ of graph $G$ is \EMPH{adjacent} to a subgraph $H$ of $G$ if (1) $v$ is in $G \setminus H$, and (2) there is an edge in $G$ between $v$ and some vertex in $H$.)
The invariant is clearly true when we make the initial call $\textsc{BuildTree}(\varnothing, G)$, and it is preserved throughout the algorithm: a recursive call to $\textsc{GrowBuffer}(\cS, \cX, H)$ or $\textsc{BuildTree}(\cS, H)$ is only made if $H$ is a maximal connected component induced by unassigned vertices.
Because we maintain this invariant, the procedure $\textsc{GrowBuffer}$ is well-defined.

\begin{figure}[h!]
    \centering
    \includegraphics[width=0.95\textwidth]{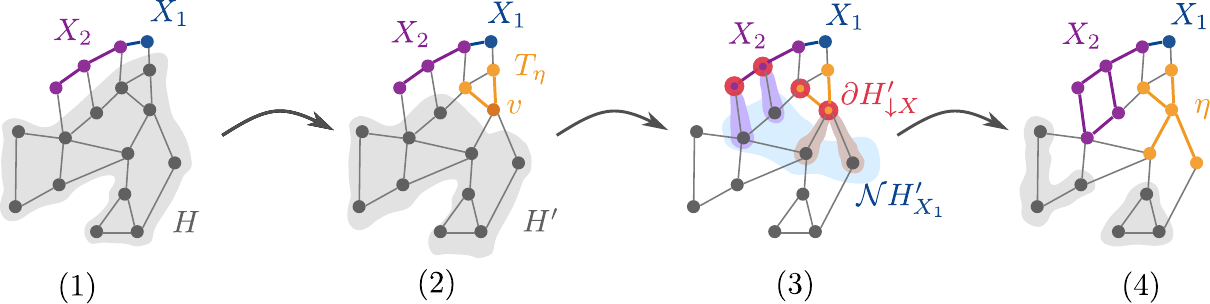}
    \caption{Example of one iteration of $\textsc{BuildTree}(H)$. From left to right: (1) The graph $G$ before an iteration, with $H$ being a connected component of unassigned vertices. (2) Pick an arbitrary vertex $v$ in $H$ and compute $T_\eta$ by taking shortest paths from $v$ to $X_1$ and $X_2$. (3) Supernode $X_1$ is cut off from $H'$, so find the set $\buff H'_{X_1}$ (vertices in $H'$ within distance $\Delta/r$ of the boundary of $\bdry H'_{\downarrow X_1}$) and assign each vertex of $\buff H'_{X_1}$ to the supernode containing the closest boundary vertex, to ensure the buffer property. (4) There are now two connected components, one of which is cut off from $X_2$, so we must grow a buffer for $X_2$ in that component.}
    \label{fig:example-algo}
\end{figure}

\paragraph{A remark on the global variable.} 
The procedure $\textsc{BuildTree}(\cS, H)$ is recursive: it initializes a super\-node, calls $\textsc{GrowBuffer}$, and then recursively calls $\textsc{BuildTree}(\cS, H'_i)$ on disjoint subgraphs $H'_i$.
Each of these recursive calls modifies the same global variable~$\cS$.
However, the modifications to $\cS$ that are made by each call $C_i \coloneqq \textsc{BuildTree}(\cS, H'_i)$ \emph{do not} affect the 
execution of any \emph{sibling} calls $C_j \coloneqq \textsc{BuildTree}(\cS, H'_j)$. Only the ancestors of $C_i$ in the recursion tree affect the execution of $C_i$.

Before proving this observation, we introduce the following important terminology. We say that a call $C \coloneqq \textsc{BuildTree}(\cS, H)$ occurs \EMPH{above} (resp.\ \EMPH{below}) 
a call $C' \coloneqq \textsc{BuildTree}(\cS, H')$
if $C$ is an ancestor (resp.\ descendent) of $C'$ in the recursion tree. If two calls to $\textsc{BuildTree}$ are in different branches of the recursion tree, then they are neither above nor below each other. Intuitively, ``$C$ is above $C'$'' means that $C$ was a \emph{relevant} call that happened \emph{before} $C'$.
Similarly, we say ``supernode $\eta_1$ is initialized \EMPH{above} supernode $\eta_2$'' if the instance of $\textsc{BuildTree}$ that initialized $\eta_1$ occurred above the instance of $\textsc{BuildTree}$ that initialized $\eta_2$.
Finally, we say that ``supernode $\eta$ is initialized \EMPH{above} a call $C \coloneqq \textsc{GrowBuffer}(\cS, \cX, H)$'' if the call to $\textsc{BuildTree}$ that initialized $\eta$ is above \emph{or is the same as}
the call to $\textsc{BuildTree}$ that caused $C$ to be called. Note that the algorithm never initializes a new supernode during a $\textsc{GrowBuffer}$ call.

\medskip
    With this terminology, we can state a stronger version of~\Cref{I:adj-assignment-weak}.
    \begin{invariant}
\label{I:adj-assignment-strong}
    Suppose that call $C$, whether it is $\textsc{GrowBuffer}(\cS, \cX, H)$ or $\textsc{BuildTree}(\cS, H)$, is made at some point in the algorithm. 
    At the time call $C$ is made, every vertex in $H$ is unassigned, and every vertex in $G \setminus H$ that is adjacent to $H$ was already assigned \ul{during a call above $C$} to some supernode \ul{initialized above $C$}.
\end{invariant}
    \noindent
    Indeed, when some call $\tilde{C}$ (whether it is \textsc{GrowBuffer} or \textsc{BuildTree}) makes call $C$ on some subgraph $H$, the graph $H$ is a maximal connected component of unassigned vertices --- and crucially, this connected component is with respect to the assignments $\cS$ \emph{before any calls to \textsc{GrowBuffer} or \textsc{BuildTree} are made by $\tilde{C}$.} Thus, every vertex adjacent to $H$ has been assigned to some existing supernode (before any sibling calls of $C$ are made), meaning it was initialized above $C$.

    \medskip
    We now prove that the execution of a call $C$, whether it is $\textsc{GrowBuffer}(\cS, \cX, H)$ or $\textsc{BuildTree}(\cS, H)$, depends only on $H$ and the vertices assigned to supernodes above $C$.
    Indeed, a call to $\textsc{BuildTree}(\cS, H)$ uses $\cS$ to determine the supernodes seen by $H$, which is determined by the vertices adjacent to $H$ (and thus, by~\Cref{I:adj-assignment-strong}, by calls above $C$ and not by sibling calls). A call to $\textsc{GrowBuffer}(\cS, \cX, H)$ uses $\cS$ in two places. First, it uses $\cS$ to determine $\bdry H_{\downarrow X}$ (where $X$ denotes the supernode selected from $\cX$ to be processed during the execution of $C$), which is a subset of the vertices adjacent to $H$ (and thus determined by calls above $C$).
    Second, it uses $\cS$ to determine, for every vertex $v$ in $H$, a closest vertex $\bdry H_{\downarrow X}$ with respect to $\dom_{\cS}(X)$.
    Notice that a shortest path $P$ in $\dom_{\cS}(X)$ from $v$ to $\bdry H_{\downarrow X}$ is contained in $H \cup \bdry H_{\downarrow X}$: the first vertex along $P$ that leaves $H$ is in $\bdry H_{\downarrow X}$, and thus is an endpoint of $P$. This means that $P$ is determined by the graph $H \cup \bdry H_{\downarrow X}$ (and as argued earlier, \Cref{I:adj-assignment-strong} implies that $\bdry H_{\downarrow X}$ is determined by calls above $C$).

This shows that
calls only affect each other if one is above the other; sibling calls do not affect each other. We will not explicitly use this fact in our proofs, instead depending solely on \Cref{I:adj-assignment-strong} --- but it is nevertheless important intuition.

\begin{figure}[htbp]
\centering
\small
\begin{algorithm}
\textul{$\textsc{BuildTree}(\cS, H)$:}\+
\\  \Comment{Initialize supernode $\eta$}
\\  $\seen{H} \gets$ supernodes in $\cS$ seen by $H$
\\  $v \gets$ arbitrary vertex in $H$
\\  $T_\eta \gets$ SSSP tree in $H$, connecting $v$ to a witness vertex for every supernode in $\seen{H}$
\\  initialize supernode $\eta \gets$ subgraph induced by vertices of $T_\eta$
\\ set $T_\eta$ to be the skeleton of $\eta$, and add $\eta$ to $\cS$
\Comment{Currently, $\dom_\cS(\eta) = H$}
\\  initialize tree $\cT$ with root $\eta$
\\
\\  \Comment{Grow buffer and recurse}
\\  for each connected component $H'$ of $H \setminus \eta$:\+
\\      $\cX \gets$ list containing supernodes seen by $H$ but not by $H'$
\\      $\textsc{GrowBuffer}(\cS, \cX, H')$
\-
\\  for each connected component $H'$ of $H \setminus \bigcup \cS$:\+
\\      $\cT' \gets \textsc{BuildTree}(\cS, H')$
\\      attach tree $\cT'$ as a child to the root of $\cT$\-
\\  return tree $\cT$
\end{algorithm}
\caption{Pseudocode for procedure $\textsc{BuildTree}(\cS, H)$}
\label{alg:buildtree}
\end{figure}

\begin{figure}[htbp]
\centering
\small
\begin{algorithm}
\textul{$\textsc{GrowBuffer}(\cS, \cX, H)$:}\+ 
\\ \Comment{Note: All supernodes in $\cX$ are not seen by $H$}
\\  If $\cX$ is the empty list, do nothing and return.
\\
\\  \Comment{Grow buffer around a supernode in $\cX$}
\\  $X \gets$ arbitrary supernode in $\cX$ \Comment{$X$ is \EMPH{processed} during this call} 
\\  $\bdry H_{\downarrow X} \gets$ set of vertices in $G \setminus H$ that are (1) adjacent to $H$, and (2) in $\dom_{\cS}(X)$
\\  $\buff H_X \gets$ vertices $v$ in $H$ such that $\dist_{\dom_\cS(X)}(v, \bdry H_{\downarrow X}) \le \Delta/r$ 
\\  \Comment{We show that $\dist_{\dom_\cS(X)}(v, \bdry H_{\downarrow X}) < \dist_{\dom_\cS(X)}(v, X)$
}
\\  for each $v$ in $\buff H_X$:\+
\\      $\eta_v \gets$ \makecell[lt]{%
supernode containing the vertex $x$ in $\bdry H_{\downarrow X}$ that minimizes $\dist_{\dom(X)}(v, x)$, \\ breaking ties consistently}
\\      assign $v$ to supernode $\eta_v$, and update $\cS$
\\      \Comment{This changes $\eta_v$, and the domains of supernodes initialized below $\eta_v$}
\\      \Comment{Since $H$ does not change, $\bdry H_{\downarrow X}$ and $\buff H_X$ remain fixed}
\-
\\
\\  \Comment{Growing a buffer may cut off more supernodes, so update $\cX$}
\\  for each connected component $H'$ of $H \setminus \bigcup \cS$:\+
\\      $\cX' \gets \cX \setminus \set{X}$ 
\\      add to $\cX'$ all supernodes in $\cS$ that are seen by $H$ but not by $H'$
\\      $\textsc{GrowBuffer}(\cS, \cX', H')$ \-
\end{algorithm}
\caption{Pseudocode for procedure $\textsc{GrowBuffer}(\cS, \cX, H)$}
\label{alg:growbuffer}
\end{figure}

\subsection{Analysis: Basic properties}

Let $\cT$ be the tree produced by $\textsc{BuildTree}(\varnothing, G)$. 
We will show that if $G$ excludes a $K_r$-minor, then $\cT$ is a $(\Delta, \Delta/r, r-1)$-buffered cop decomposition for $G$. In this section, we prove a collection of basic properties about $\cT$, including the shortest-path skeleton and tree decomposition properties. The proofs of the supernode buffer and supernode radius properties are deferred to the next two sections.

\paragraph{Notation for supernodes changing over time.} 
When we write ``supernode $\eta$'' without any subscript or description, we refer to the supernode in $\cT$, at the end of the execution of the entire algorithm. In some proofs, we will need to refer to the global variable $\cS$ at a specific point in the algorithm's execution. We adopt the following convention: If we say a call $C' \coloneqq \textsc{BuildTree}(\cS', H')$ is made during the algorithm, we use the variable $\cS'$ to denote the global variable at the \emph{start} of call $C'$.
We use the notation ``\EMPH{supernode $\eta_{\cS'}$}'' to refer to the vertices of $G$ that have already been assigned to $\eta$ by $\cS'$. 
It does \emph{not} refer to those vertices that will be assigned to $\eta$ in the future. 
The phrase ``the set of supernodes \EMPH{in $\cS'$}'' refers to the set of supernodes $\eta_{\cS'}$ that are assigned by $\cS'$. 

\paragraph{Terminology for $\textsc{GrowBuffer}$.} Suppose that some call $C \coloneqq \textsc{GrowBuffer}(\cS, \cX, H)$ occurs during the algorithm. This call begins by selecting an arbitrary supernode $X$ from $\cX$; we say that $X$ is the supernode \EMPH{processed during $C$}. 
The call $C$ then defines the set $\buff H_X$; we say that every point in $\buff H_X$ is \EMPH{assigned during $C$}.

\begin{claim}
\label{C:basic}
The following properties hold.
\begin{enumerate}
    \item \label{C:connected} For every $\cS$ that appears in the algorithm, every supernode $\eta_{\cS}$ induces a connected subgraph of $G$.
    \item \label{clm:future-assignment}
    Suppose that call $C$, whether it is $\textsc{GrowBuffer}(\cS, \cX, H)$ or $\textsc{BuildTree}(\cS, H)$, is made at some point in the algorithm.
    Over the course of the algorithm, every vertex in $H$ is assigned either to a supernode initialized by $C$, or to a supernode initialized below $C$, or to a supernode in $\cS$ that $H$ sees (at the time $C$ is called).

    \item \label{clm:grow-n-shrink}
    Supernode $\eta_{\cS}$ will grow and $\dom_{\cS}(\eta)$ will shrink over the course of the algorithm as $\cS$ changes. Further, after the algorithm terminates, we have $\dom_{\cS}(\eta) = \dom(\eta)$.
\end{enumerate}
\end{claim}

\begin{proof}  
    \textbf{(1)} When supernode $\eta$ is initialized, it is connected (because the skeleton $T_\eta$ is connected). Whenever a vertex $v$ is assigned to $\eta$ by a call to $\textsc{GrowBuffer}(\cS, \cX, H)$, we claim that connectivity is preserved. 
    Let $X$ denote the supernode processed during $\textsc{GrowBuffer}(\cS, \cX, H)$, let $\bdry H_{\downarrow X}$ denote the set of boundary vertices, and let $\buff H_X$ denote the vertices assigned during $\textsc{GrowBuffer}(\cS, \cX, H)$. 
    Let $P$ be a shortest path in $\dom_{\cS}(X)$ from $v$ to the closest point $\bdry H_{\downarrow X}$. Every vertex in $P$ is closer to $\bdry H_{\downarrow X}$ than $v$. Further, we claim that every vertex in $P$ (excluding the endpoint, which is a boundary vertex) is in $H$.
    Indeed, every vertex in $P$ is in $\dom_{\cS}(X)$, so the first vertex along $x$ that leaves $H$ is in $\bdry H_{\downarrow X}$, and thus is the endpoint of $P$.
    Thus, every point in $P$ (excluding the endpoint) is in $\buff H_X$.
    As every vertex in $\buff H_X$ is assigned according to the closest vertex in $\bdry H_{\downarrow X}$ (and ties are broken consistently), every vertex in path $P$ is assigned to the same supernode $\eta$, and the connectivity of $\eta$ is preserved.

    \medskip\noindent\textbf{(2)}
    Let $v$ be a vertex in $H$ assigned to supernode $\eta$. Suppose that $\eta \subseteq H$ (and suppose that $C$ itself did not initialize $\eta$). In this case, we claim that $\eta$ was initialized below $C$. This follows from the fact that, for any call to \textsc{BuildTree}, the children calls to \textsc{BuildTree} in the recursion tree operate on pairwise disjoint subgraphs. This implies that $\eta$ is initialized either above or below $C$; as $H$ consists of unassigned nodes at the time $C$ is called, $\eta$ is initialized below $C$.

    Now suppose that $\eta$ is not contained in $H$. As $\eta$ is connected (\Cref{C:connected}), there is a path $P$ in $\eta$ between $v$ and a vertex outside of $H$. By \Cref{I:adj-assignment-weak}, the first vertex along $P$ that leaves $H$ has already been assigned in $\cS$, at the time $C$ is called. Thus, $H$ sees $\eta_{\cS}$.
    
    \medskip\noindent
    \textbf{(3)}
    The fact that supernodes only grow over time is immediate from the algorithm.
    Let $H$ denote the domain of $\eta$ at the time $\eta$ is initialized. By definition, $\dom_{\cS}(\eta)$ is the set of vertices in $H$ that are \emph{now} assigned to supernodes above $\eta$; as supernodes only grow over time,   $\dom_{\cS}(\eta)$ only shrinks.
    It remains to show that the final  $\dom_{\cS}(\eta)$ is equal to $\dom(\eta)$.
    Indeed, by \Cref{clm:future-assignment} (applied to the call to $\textsc{BuildTree}$ that initialized $\eta$),
    every vertex in $H$ is assigned to a supernode initialized either above or below $\eta$ (or is assigned to $\eta$ itself).
    A supernode is below $\eta$ in the partition tree $\cT$ if and only if it is initialized below $\eta$. Thus, after the algorithm terminates, $\dom_{\cS}(\eta)$ is the set of vertices that are in $\eta$ or in supernodes below $\eta$, which is precisely $\dom(\eta)$.
\end{proof}

\begin{lemma}
\label{lem:seen-width}
    Suppose that $\textsc{BuildTree}(\mathcal{S}, H)$ is called during the algorithm.
    Let $\seen{H}$ be the set of supernodes in $\cS$ seen by $H$. 
    Then $\seen{H}$ contains at most $r-2$ supernodes; furthermore, the supernodes in $\seen{H}$ are pairwise adjacent.
\end{lemma}

\begin{proof}
    We first prove that the supernodes in $\seen{H}$  are pairwise adjacent. 
    Consider a pair $(X, Y)$ 
    of supernodes 
    in $\seen{H}$, and assume without loss of generality%
    \footnote{By \Cref{I:adj-assignment-strong}, both $X$ and $Y$ are initialized above $\textsc{BuildTree}(\cS, H)$, so one of $X$ or $Y$ was initialized below the other.} that $Y$ is initialized below $X$.
    Let $x$ and $y$ be the vertices chosen to be the roots of the skeletons of $X$ and $Y$, respectively. 
    Since $H$ sees both $X_{\cS}$ and $Y_{\cS}$, and as $X_{\cS}$ and $Y_{\cS}$ are connected individually, there exists a path $P$ from $x$ to $y$ containing only vertices in $H$, $X_{\cS}$, and $Y_{\cS}$. 
    
    Consider the time just before $Y$ is initialized. Let $\tilde{\cS}$ denote the assignments of vertices at that time, and let $\tilde{H}$ be the connected component of the subgraph induced by the unassigned vertices that contains $y$ at that time. Observe that,
    although $X_{\tilde{\cS}}$ may expand later,
    all vertices in $P$ are either unassigned (and thus belong to $\tilde{H}$) or belong to $X_{\tilde{\cS}}$. 
    Hence, there is a path from $y$ to $X_{\tilde{\cS}}$ containing only unassigned vertices, 
    meaning that $\tilde{H}$ sees $X_{\tilde{\cS}}$. 
    Thus, when $Y_{\tilde{\cS}}$ is initialized, it must be adjacent to $X_{\tilde{\cS}}$ by construction.
    By \Cref{C:basic}(\ref{clm:grow-n-shrink}), supernodes only grow, and thus $X$ and $Y$ must be adjacent.
    
    By \Cref{C:basic}(\ref{C:connected}), every supernode is connected. Thus, the above claim implies that $\seen{H} \cup \set{\eta}$ forms a model for a $K_{k + 1}$-minor, where $k = \abs{\seen{H}}$. As $G$ excludes $K_r$-minors, we have $\abs{\seen{H}} \le r - 2$.
\end{proof}

\begin{lemma}[Shortest-path skeleton property]
\label{lem:shortest-path}
    Every supernode $\eta$ has a skeleton that is an SSSP tree in $\dom(\eta)$ with $r-2$ leaves.
\end{lemma}

\begin{proof}
    Notice that, throughout the course of the algorithm, $\dom_{\mathcal{S}}(\eta)$ may shrink but it never expands by \Cref{C:basic}(\ref{clm:grow-n-shrink}).
    As $T_\eta$ is an SSSP tree in its original domain $\dom_{\mathcal{S}}(\eta)$, it is also an SSSP tree in $\dom(\eta)$.  (We remark that $\eta$ only grows over the course of the algorithm, so $T_\eta$ is a subgraph of $\eta$, and thus is a subgraph of $\dom(\eta)$).
    By \Cref{lem:seen-width}, tree $T_\eta$ has at most $r-2$ leaves.
\end{proof}

\begin{claim}
\label{clm:see-nesting}
    Let $C$ be a call, whether to $\textsc{BuildTree}(\cS, H)$ or to $\textsc{GrowBuffer}(\cS, \cX, H)$, and let $\eta$ be a supernode initialized above $C$. Let $H'$ be a subgraph of $H$. If $H'$ is adjacent to a vertex in $\eta$ (after the algorithm terminates), then $H$ sees $\eta_{\cS}$.
\end{claim}

\begin{proof}
    Let $v$ be a vertex in $\eta \setminus H'$ adjacent to $H'$. As $\eta$ is connected, there is a path $P$ from $v$ to $T_\eta$ contained in $\set{v} \cup \eta$. As $\eta$ is initialized above $C$ (and, at the time $C$ is called, subgraph $H$ contains only unassigned vertices), the skeleton $T_\eta$ is disjoint from $H$. Thus, $P$ starts at a vertex in $H$ and ends at a vertex outside of $H$. Consider the first vertex along $P$ that leaves $H$. By \Cref{I:adj-assignment-weak}, this vertex had already been assigned (to $\eta$) at the time $C$ was called. We conclude that $H$ sees $\eta_{\cS}$.
\end{proof}

\begin{lemma}[Tree decomposition property]
\label{lem:tree-decomposition}
    $\hat{\cT}$ satisfies the tree decomposition property.
\end{lemma}

\begin{proof}
First, note that \Cref{lem:seen-width} directly implies that each supernode sees at most $r-2$ ancestor supernodes, meaning that each bag in $\hat{T}$ contains at most $r-1$ supernodes. 
Next, we prove that $\hat{\cT}$ is a tree decomposition. 
    
    \medskip\noindent
    \textbf{(1)} The union of all vertices in all bags of $\hat{\cT}$ is $V$ by construction. 
        
    \medskip\noindent
    \textbf{(2)} Let $(x,y)$ be an edge in $G$. We need to show that there is a bag in $\hat{\cT}$ that contains both $x$ and $y$. Let $X$ be the supernode containing $x$, and let $Y$ be the supernode containing $y$. We will prove that either $X = Y$ or one of them is an ancestor of the other (recall that, by definition, the bag of $X$ contains all supernodes above $X$ that are adjacent to $X$).
    
    Assume that $X \neq Y$. We claim that $X$ and $Y$ are in an ancestor-descendent relationship in $\cT$.  
    Otherwise, consider the lowest common ancestor $\eta$ of $X$ and $Y$, initialized by a call $C \coloneqq \textsc{BuildTree}(\cS, H)$. 
    As $X$ and $Y$ are in different subtrees of $\eta$, vertices $x$ and $y$ are both unassigned and belong to different connected components of unassigned vertices, at the time when $C$ begins to recursively make calls to $\textsc{BuildTree}$. But this is impossible, as there is an edge between $x$ and $y$.

    \medskip\noindent
    \textbf{(3)} We prove that for any supernode $\eta$, if there are two bags $\bag{X}$ and $\bag{Y}$ containing $\eta$, every bag in the path between them in $\hat{\cT}$ contains $\eta$.
    
    Let $P$ be the path between $\bag{X}$ and $\bag{Y}$ in $\hat{\cT}$.
    Assume that there exists some bag in $P$ not containing $\eta$.
    Observe that the bag $\bag{\eta}$ is a common ancestor of both $\bag{X}$ and $\bag{Y}$. 
    Consider two paths: $P_X$ from $\bag{\eta}$ to $\bag{X}$ and $P_Y$ from $\bag{\eta}$ to $\bag{Y}$. 
    One of them, say $P_X$, must have a bag that does not contain $\eta$. 
    Let $\bag{\eta'}$ be the lowest bag in $P_X$ such that $\bag{\eta'}$ does not contain $\eta$, and let $\bag{\eta''}$ be the child of $\bag{\eta'}$ in $P_X$. Notice that $\bag{\eta''}$ contains $\eta$. We remark that $\bag{\eta'}$ is a descendent of $\bag{\eta}$.
    From the construction of $\hat{\cT}$, we get that supernode $\eta''$ is adjacent to $\eta$ but supernode $\eta'$ is not. 
    Suppose that $\eta'$ is initialized during the call $C' \coloneqq \textsc{BuildTree}(\cS', H')$, and $\eta''$ is initialized during the call $C'' \coloneqq \textsc{BuildTree}(\cS'', H'')$.
    As $\eta$ is initialized above $C'$, and $H'' \subseteq H'$, and $H''$ is adjacent to $\eta$, \Cref{clm:see-nesting} implies that $H'$ sees $\eta_{\cS'}$ at the time $C'$ is called. Thus, by construction $\eta'$ is adjacent to $\eta$, a contradiction. 
\end{proof}

\subsection{Analysis: Supernode buffer property}

The following observation is almost immediate from the construction. It says that, if some subgraph $H'$ is cut off from an old supernode $X$, there was some call to $\textsc{GrowBuffer}$ that processed $X$.
\begin{observation}
\label{obs:processing-call}
    Suppose that call $C' \coloneqq \textsc{BuildTree}(\cS', H')$ is made during the algorithm. If $X$ is a supernode initialized above $C'$, and if $H'$ does not see $X_{\cS'}$ at the time $C'$ is called, then there is some call $C \coloneqq \textsc{GrowBuffer}(\cS, \cX, H)$ such that (1) $H \supseteq H'$, and in particular $C$ is above $C'$, (2) $H$ does not see $X_{\cS}$, and (3) $X$ was processed during $C$.
\end{observation}

To see why the observation holds, denote by $\tilde{C} \coloneqq \textsc{BuildTree}(\tilde{\cS}, \tilde{H})$ the lowest call above $C'$ such that $\tilde{H}$ sees $X_{\tilde{\cS}}$ (or, if no such call exists, let $\tilde{C}$ be the call that initializes $X$). After making some calls to $\textsc{GrowBuffer}$, the call $\tilde{C}$ must recurse on some subgraph that does not see $X$. Since the algorithm calls $\textsc{GrowBuffer}$ whenever a supernode gets cut off, there must be some (recursive) call to $\textsc{GrowBuffer}$ caused by $\tilde{C}$ that processed $X$, as claimed by the observation.
We shall use this observation to prove the supernode buffer property.

\begin{lemma}[Supernode buffer property]
\label{lem:buffer}
Let $X$ and $\eta$ be supernodes, with $\eta$ initialized below $X$. If $\eta$ is not adjacent to $X$ in $G$, then for every vertex $v$ in $\dom(\eta)$, we have $\dist_{\dom(X)}(v, X) > \Delta/r$.
\end{lemma}

\begin{proof}
    We prove the following claim by induction (starting immediately below the call that initialized $X$, and working downward in the recursion tree):
    \begin{quote}
        Let $C' \coloneqq \textsc{BuildTree}(\cS', H')$ be a call that is below the call that initialized $X$. Either $H'$ sees $X_{\cS'}$, or  $\dist_{\dom(X)}(v, X) > \Delta/r$ for every vertex $v$ in $H'$.
    \end{quote}
    We emphasize that the guarantee $\dist_{\dom(X)}(v, X) > \Delta/r$ refers to the \emph{final} $X$, after all expansions are made.
    This suffices to prove the lemma. 
    Indeed, the call to $\textsc{BuildTree}(\cS', H')$ that initialized $\eta$ comes below the call that initialized $X$, and $\dom(\eta) \subseteq H'$ by \Cref{C:basic}(\ref{clm:grow-n-shrink}); thus, either $H'$ sees $X_{\cS'}$ (in which case $\eta$ is adjacent to $X$ by definition of $T_\eta$), or every point $v$ in $\dom(\eta)$ satisfies $\dist_{\dom(X)}(v, X) > \Delta/r$.

    \bigskip \noindent \textbf{Inductive step.} \, 
    Suppose that $H'$ does not see $X_{\cS'}$.
    As we are in the inductive step, we may assume that the parent of $C'$ in the recursion tree, $\tilde{C} \coloneqq \textsc{BuildTree}(\tilde{\cS}, \tilde{H})$, is below the call that initialized $X$.
    If $\tilde{H}$ does not see $X_{\tilde{\cS}}$, then we are done:  
    since graph $\tilde{H}$ is a supergraph of $H'$, the inductive hypothesis implies that $\dist_{\dom(X)}(v, X) > \Delta/r$ for every vertex $v$ in $H'$.

    The interesting case occurs when $\tilde{H}$ sees $X_{\tilde{\cS}}$, but $H'$ does not see $X_{\cS'}$: that is, $X$ becomes ``cut off'' from $H'$ some time in between. 
    In this case, by \Cref{obs:processing-call} there is some call
    $C \coloneqq \textsc{GrowBuffer}(\cS, \mathcal{X}, H)$ above $C'$ that processes $X$, with $H \supseteq H'$ and $H$ does not see $X_{\cS}$. 
    Consider any vertex $v$ in $H$ such that $\dist_{\dom(X)}(v,X) \le \Delta/r$ (where, again, we emphasize that $X$ refers to the \emph{final} $X$, after all expansions). If no such vertex exists, we are done.

    \begin{figure}[h!]
        \centering
        \includegraphics[width=\textwidth]{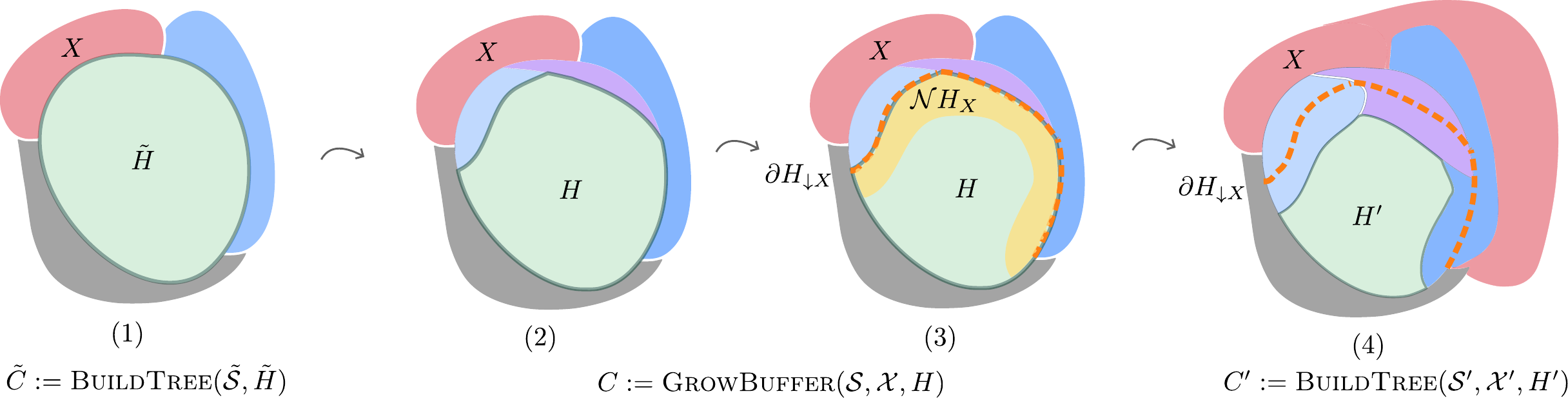}
        \caption{From left to right: (1) During call $\tilde{C}$, subgraph $\tilde{H}$ sees supernode $X$. The grey supernode is \emph{above} $X$, and is not in $\dom_{\tilde{S}}(X)$. 
        (2--3) During call $C$, supernode $X$ is cut off from $H$, and every point in $\buff H_X$ (i.e.\ every point close to $\bdry H_{\downarrow X}$) is assigned.  
        (4) For every subgraph $H'$ of $H$, every path in $\dom(X)$ from $H'$ to $X$ passes through $\bdry H'_{\downarrow X}$.}
        \label{fig:supernode-buffer}
    \end{figure}

    We argue that \emph{vertex $v$ is at most $\Delta/r$ away from $\bdry H_{\downarrow X}$ with respect to $\dom_{\cS}(X)$}; see~\Cref{fig:supernode-buffer}.
    Let $P$ be a shortest path from $v$ to $X$ in $\dom(X)$, where by assumption $\len{P} \le \Delta/r$.  
    As the domain of $X$ only shrinks over time (\Cref{C:basic}(\ref{clm:grow-n-shrink})), path $P$ is in $\dom_{\cS}(X)$.%
    \footnote{However, notice that at the time when call $C$ was made, $X$ might not have grown into its final shape and $X_\cS$ could be much smaller; in particular, $P$ may not be a path from $v$ to $X_\cS$ and the distance from $v$ to $X_\cS$ can be larger than $\Delta/r$.}
    By \Cref{C:basic}(\ref{clm:future-assignment}) on the call $C$, 
    every vertex in $H$ is assigned either to a supernode initialized below $C$, or to a supernode in $\cS$ that $H$ sees.
    Because $X$ already existed in $\cS$ (and thus $X$ is not initialized below $C$) and $H$ does not see $X_{\cS}$, the other endpoint of $P$ which is eventually assigned to $X$ cannot be in $H$. 
    So $P$ passes through some vertex $x$ outside of $H$ that is adjacent to $H$. 
    As $P$ is contained in $\dom_{\cS}(X)$, vertex $x$ is in $\bdry H_{\downarrow X}$.
    Thus, $\dist_{\dom_{\cS}(X)}(v, \bdry H_{\downarrow X}) \le \len{P} \le \Delta/r$.

    This means that $v$  is assigned to some supernode in $\cS$ by the $\textsc{GrowBuffer}$ algorithm.
    Recall that call $C'$ is below call $C$ by~\Cref{obs:processing-call}; thus, as calls are only made on connected components of unassigned vertices, we conclude that  $H'$ is a subgraph of $H$ that includes only unassigned vertices.  Thus, vertex $v$ is not in $H'$. This completes the proof of the induction step.

    \medskip \noindent \textbf{Base case.} \,
    In the base case, the parent of $C'$ in the recursion tree is the call $\tilde{C} \coloneqq \textsc{BuildTree}(\tilde{\cS}, \tilde{H})$ that initialized $X$. If $H'$ sees $X_{\cS'}$, then we are done. If $H'$ does not see $X_{\cS'}$, then we are in the ``interesting case'' described above (except that here 
    $\tilde{H}$ doesn't see $X_{\tilde{\cS}}$): by~\Cref{obs:processing-call}, there is some call $C \coloneqq \textsc{GrowBuffer}(\cS, \cX, H)$ above $C'$ and below $\tilde{C}$, during which $X$ is processed. The argument for this case is identical to the one in the inductive case.
\end{proof}

\subsection{Analysis: Supernode radius property}
We now prove that every supernode $\eta$ satisfies the radius property. To this end we prove three claims: 
\begin{itemize}
\item Every time a supernode is cut off from a subgraph, the radius of $\eta$ expands by at most $\Delta/r$ (\Cref{clm:expansion}). 
\item There are at most $r - 2$
supernodes that can cause $\eta$ to expand
(\Cref{clm:valid-expansion}). 
\item Each of the $r-2$ supernodes can cause $\eta$ to expand at most once (\Cref{clm:spent}).
\end{itemize}
Combining these three claims in an inductive argument shows that the total expansion of $\eta$ is bounded by~$\Delta$ (\Cref{lem:radius}).

\begin{claim}
    \label{clm:expansion}
    Suppose that $v$ is assigned to a supernode $\eta$ during a call $C \coloneqq \textsc{GrowBuffer}(\cS, \cX, H)$. 
    Let $X$ denote the supernode processed during $C$, and let $\bdry H_{\downarrow X}$ denote the boundary vertices. Let $\tilde{v}$ be the closest vertex in $\bdry H_{\downarrow X}$ 
    to~$v$ (with respect to $\dom_{\cS}(X)$).
    Then $\dist_{\eta}(v, \tilde{v}) \le \Delta/r$ (with respect to the final $\eta$).
\end{claim}
\begin{proof}
    Let $\buff H_X$ denote the set of points assigned during $C$.
    Let $P$ be a shortest path between $v$ and $\tilde{v}$ in $\dom_{\cS}(X)$. 
    Every vertex in $P$ (other than $\tilde{v}$) is in $\buff H_X$. Because we assign every vertex in $\buff H_X$ according to the closest vertex in $\bdry H_{\downarrow X}$, every vertex in $P$ is assigned to $\eta$. Further, $P$ has length at most $\Delta/r$, because every vertex in $\buff H_X$ is within distance $\Delta/r$ of some vertex in $\bdry H_{\downarrow X}$ (with respect to $\dom_\cS(X)$), and $\tilde{v}$ is the closest vertex in $\bdry H_{\downarrow X}$ to~$v$. 
    Thus, $\dist_{\eta}(v, \tilde{v}) \le \Delta/r$.
\end{proof}
    
We next show that each supernode seen by supernode $\eta$ may cause $\eta$ 
to be expanded at most once: if supernode $\tilde{X}$ causes $\eta$ to expand because $\tilde{X}$ is cut off, supernode $\tilde{X}$ cannot be cut off \emph{again} later on in the recursion. 
Later (in \Cref{clm:valid-expansion}) we will show that only supernodes seen by $\eta$ may cause it to expand.
Let $\tilde{X}$ be a supernode, and let $H$ be a subgraph.
We say that $\tilde{X}$ is \EMPH{spent}
with respect to $H$ if there exists some call $\textsc{GrowBuffer}(\tilde{\cS}, \tilde{\cX}, \tilde{H})$ where $\tilde{H} \supseteq H$, and $\tilde{X}$ is processed during the call. 
In other words, $\tilde{X}$ is cut off from $\tilde{H}$ and $H$ (even as $\tilde{X}$ grows), and it has already been ``dealt with'' during the previous call $\tilde{C}$.

\begin{claim}
\label{clm:spent}
    Suppose that call $\textsc{GrowBuffer}(\cS, \cX, H)$ is made during the algorithm. 
    If supernode $\tilde{X}$ is spent with respect to $H$, then $\tilde{X}$ is not in $\cX$.
\end{claim}

\begin{proof}
    By definition of ``spent'', there is some call $\tilde{C} \coloneqq \textsc{GrowBuffer}(\tilde{\cS}, \tilde{\cX}, \tilde{H})$ where $\tilde{H} \supseteq H$,
    and $\tilde{X}$ is processed during $\tilde{C}$. 
    Notice that, because $\tilde{X}$ is in $\tilde{\cX}$, subgraph $\tilde{H}$ does not see $\tilde{X}_{\tilde{\cS}}$.  Observe that:
    
\begin{itemize}
    \item Call $\tilde{C}$ makes some calls to $\textsc{GrowBuffer}(\cS', \cX', H')$. For each of these calls made by $\tilde{C}$, notice that the set $\cX'$ contains only supernodes in $\tilde{\cX} \setminus \set{\tilde{X}}$ (the ``leftover'' ones from $\tilde{C}$) or supernodes in $\tilde{\cS}$ that can be seen by $\tilde{H}$ but not by $H'$ (those newly added ones). In particular, $\cX'$ does not contain $\tilde{X}$. Further, \emph{$H'$ does not see $\tilde{X}_{\cS'}$}. This follows from \Cref{clm:see-nesting}: if $H'$ could see $\tilde{X}_{\cS'}$, then (because $\tilde{X}$ was initialized above $\tilde{C}$, and $H' \supseteq \tilde{H}$), \Cref{clm:see-nesting} would imply that $\tilde{H}$ could see $\tilde{X}_{\tilde{\cS}}$, a contradiction.
    
    An inductive argument shows that, for every call to $\textsc{GrowBuffer}(\cS', \cX', H')$ made recursively as a result of $\tilde{C}$, the set $\cX'$ does not contain $\tilde{X}$.
     
    \item After the recursion from $\tilde{C}$ terminates, the algorithm calls $\textsc{BuildTree}$ on subgraphs of $\tilde{H}$, which may recursively result in more calls to $\textsc{BuildTree}$.  Let $\textsc{BuildTree}(\cS', H')$  be one of these calls, where $H' \subseteq \tilde{H}$.
    We claim that \emph{$H'$ does not see $\tilde{X}_{\cS'}$}. As in the previous bullet point, this follows from \Cref{clm:see-nesting}: if $H'$ could see $\tilde{X}_{\cS'}$, then \Cref{clm:see-nesting} would imply that $\tilde{H}$ could see $\tilde{X}_{\tilde{\cS}}$, a contradiction.

    This means that whenever $\textsc{BuildTree}(\cS', H')$ makes a call $\textsc{GrowBuffer}(\cS'', \cX'', H'')$, the set $\cX''$ does not include $\tilde{X}$; indeed, the set $\cX''$ only includes supernodes that are seen by $H'$.
\end{itemize}

\noindent
It follows from these two cases that, for every call to $\textsc{GrowBuffer}(\cS', \cX', H')$ with $H' \subseteq \tilde{H}$, the supernode $\tilde{X}$ is not in $\cX'$. In particular, the call $\textsc{GrowBuffer}(\cS, \cX, H)$ satisfies $H \subseteq \tilde{H}$, and so $\tilde{X}$ is not in $\cX$.
\end{proof}

The following claim, in conjunction with  \Cref{lem:seen-width}, implies that for any supernode $\eta$, at most $r-2$ supernodes can cause it to expand. We crucially rely on the fact that when supernode $X$ is cut off, we only expand supernodes initialized \emph{below} $X$; we do this because we only need to guarantee the supernode buffer property with respect to $\dom(X)$.
\begin{claim}
    \label{clm:valid-expansion}
    Suppose that $v$ is assigned to supernode $\eta$ during a call $C \coloneqq \textsc{GrowBuffer}(\cS, \cX, H)$, and let $X$ be the supernode in $\cX$ processed during $C$. Suppose that $\eta$ was initialized by a call $\hat{C} \coloneqq \textsc{BuildTree}(\hat{\cS}, \hat{H})$. Then $\hat{H}$ sees $X_{\hat{\cS}}$.
\end{claim}

\begin{proof}
    We first show that $\eta$ is initialized below $X$. As $v$ is assigned to supernode $\eta$ during $C$, there is some vertex in $\bdry H_{\downarrow X} \subseteq \dom_{\cS}(X)$ that was assigned to $\eta$ (in $\cS$). Suppose that $X$ was initialized during the call $\bar{C} \coloneqq \textsc{BuildTree}(\bar{\cS}, \bar{H})$, and notice that $\dom_{\cS}(X) \subseteq \bar{H}$. Applying~\Cref{C:basic}(\ref{clm:future-assignment}) to call $\bar{C}$ shows that every vertex in $\bar{H}$ is (eventually) assigned to a supernode above $X$, or to $X$, or to a supernode below $X$.
    By definition, $\dom_{\cS}(X)$ contains the vertices in $\bar{H}$ that are \emph{not} assigned to supernodes above $X$ (in $\cS$). 
    Thus, every vertex in $\dom_{\cS}(X)$ is assigned to $X$ or to a supernode below $X$, and so either $\eta = X$ or $\eta$ is below $X$.
    It cannot be that $\eta = X$ (indeed, $H$ sees $\eta$ because $v \in \bdry H_{\downarrow X}$, and $H$ does not see $X$ because $X$ is in $\cX$), so $\eta$ is below $X$.

    We also observe that $C$ is below $\hat{C}$: because $H$ is adjacent to a vertex in $\eta$, \Cref{I:adj-assignment-strong} implies that $\eta$ was initialized above $C$. Thus, $\hat{H} \supseteq H$.

    Now, for the sake of contradiction suppose that $\hat{H}$ does not see $X_{\hat{\cS}}$. As $\eta$ is below $X$ and $\hat{H}$ does not see $X_{\hat{\cS}}$, \Cref{obs:processing-call} implies that there is some call $\tilde{C} \coloneqq \textsc{GrowBuffer}(\tilde{\cS}, \tilde{\cX}, \tilde{H})$ such that (1) $\tilde{H} \supseteq \hat{H}$, and (2) $X$ is processed by $\tilde{C}$. As $\hat{H} \supseteq H$, this means that $X$ is spent with respect to $H$, and \Cref{clm:spent} implies that $X$ is not in $\cX$, a contradiction.
\end{proof}

\begin{lemma}
    \label{lem:radius}
    Every supernode $\eta$ has radius $\Delta$ with respect to skeleton $T_\eta$.
\end{lemma}

\begin{proof}
    Let $\textsc{BuildTree}(\hat{\cS}, \hat{H})$ be the call that initialized $\eta$, and let \EMPH{$\seenhat{\hat{H}}$} denote the set of supernodes in $\hat{\cS}$  that can be seen by $\hat{H}$. In other words, by \Cref{clm:valid-expansion}, $\seenhat{\hat{H}}$ is the set of supernodes that can cause $\eta$ to expand.
    We prove the following statement by induction on $k$.
    
    \begin{quote}
        Let $v$ be a vertex assigned to $\eta$ during a call $C \coloneqq \textsc{GrowBuffer}(\cS, \cX, H)$. If there are at most $k$ supernodes in $\seenhat{\hat{H}}$ that are spent with respect to $H$, then $\dist_{\eta}(v, T_\eta) \le (k + 1)\cdot \Delta/r$.
    \end{quote}

    \noindent
    Let $X$ be the supernode processed during the call $C$. 
    Let \EMPH{$\tilde{v}$} be the closest vertex to $v$ in $\eta_{\cS}$ (with respect to $\dom_{\cS}(X)$), as defined in \Cref{clm:expansion}.

    \medskip\noindent
    \textbf{Inductive step ($k > 0$).} If $\tilde{v}$ is in $T_\eta$, then \Cref{clm:expansion} implies that $v$ is within distance $\Delta/r$ of $T_\eta$ (in the final subgraph $\eta$), satisfying the claim.
    Otherwise, $\tilde{v}$ was assigned to $\eta$ by some call $\tilde{C} \coloneqq \textsc{GrowBuffer}(\tilde{\cS}, \tilde{\cX}, \tilde{H})$.
    
    We now show that the number of supernodes spent with respect to $H$ is strictly greater than the number of supernodes spent with respect to $\tilde{H}$, aiming to apply the induction hypothesis on $\tilde{H}$. 
    \begin{itemize}
    \item \emph{Every supernode in $\seenhat{\hat{H}}$ that is spent with respect to $\tilde{H}$ is also spent with respect to $H$.} 
    Indeed, for every such supernode $\bar{X}$
    spent with respect to $\tilde{H}$, there is some call $\bar{C} \coloneqq \textsc{GrowBuffer}(\bar{\cS}, \bar{\cX}, \bar{H})$ such that $\bar{X} \supseteq \tilde{X}$ and $\bar{C}$ processes $\bar{X}$; as $\tilde{H} \supseteq H$, call $\bar{C}$ also serves as a witness that $\bar{X}$ is spent with respect to $H$. 
    \end{itemize}
    Now, consider the supernode $\tilde{X}$ that was processed during $\tilde{C}$. Observe that: 
    \begin{itemize}
    \item \emph{$\tilde{X}$ is not spent with respect to $\tilde{H}$.} 
    This follows from \Cref{clm:spent} and the fact that $\tilde{X} \in \tilde{\cX}$.

    \item \emph{$\tilde{X}$ \emph{is} spent with respect to $H$.} 
    First we argue that \emph{$\tilde{H} \supsetneq H$}. 
    Observe that call $\tilde{C}$ is above call $C$, because $\tilde{v}$ was already assigned when $C$ is called. Thus, vertices that are unassigned when $C$ is called are also unassigned when $\tilde{C}$ was called. In particular, when $\tilde{C}$ was called, every vertex in $H \cup \set{\tilde{v}}$ is unassigned. As $\tilde{H}$ is a maximal connected component of unassigned vertices and $\tilde{v}$ is adjacent to $H$ by definition, $\tilde{H} \supsetneq H$.

    The existence of call $\tilde{C}$ (which processes $\tilde{X}$) together with the fact that $\tilde{H} \supseteq H$ implies that $\tilde{X}$ is spent with respect to~$H$.
    \end{itemize}
    
    \noindent 
    Moreover, \Cref{clm:valid-expansion} implies that $\tilde{X}$ is in $\seenhat{\hat{H}}$, because a vertex was assigned to $\eta$ during a call to $\textsc{GrowBuffer}$ in which $\tilde{X}$ is processed. 
    We conclude that there is at least one more supernode in $\seenhat{\hat{H}}$ that is spent with respect to $H$ than those with respect to $\tilde{H}$. 
    Thus, we can apply the inductive hypothesis to conclude that $\tilde{v}$ is within distance $k\cdot \Delta/r$ of $T_\eta$. By \Cref{clm:expansion}, $v$ is within distance $\Delta/r$ of $\tilde{v}$, and so $v$ is within distance $(k+1) \cdot \Delta/r$ of $T_\eta$.

    \medskip\noindent
    \textbf{Base case ($k=0$)}: 
    In this case, we claim $\tilde{v}$ must be in $T_\eta$, and so $\dist_{\eta}(v, T_\eta) \le \Delta/r$. 
    Suppose otherwise. 
    Then $\tilde{v}$ is assigned by a call to $\textsc{GrowBuffer}$,
    and the argument above implies that there is at least one supernode in $\seenhat{\hat{H}}$ that is spent with respect to $H$. This contradicts our assumption that $k=0$.

    \medskip \noindent
    By \Cref{lem:seen-width}, there are at most $r-2$ supernodes in $\seenhat{\hat{H}}$, and so we conclude that every vertex in $\eta$ is within distance $\Delta$ of $T_\eta$.
    \end{proof}

We conclude:
\begin{theorem}
    Let $G$ be a $K_r$-minor-free graph, and let $\Delta$ be a positive number. Then $G$ admits a $(\Delta, \Delta/r, r-1)$-buffered cop decomposition.
\end{theorem}

\section{Shortcut partition from buffered cop decomposition}
\label{sec:shortcut}

We first rephrase the definition of shortcut partition.  Let $G$ be a graph, let $\e$ be a number in $(0,1)$, and let $\cC$ be a partition of the vertices of $G$ into clusters of strong diameter $\e \cdot \diam(G)$.
Recall that the cluster graph $\check{G}$ is obtained from the original graph $G$ by contracting each cluster in $\cC$ into a single supernode. 
Let $P$ be an arbitrary path in $G$.
We define \EMPH{$\cost_{\cC}(P)$} to be the minimum hop-length of a path $\check{P}$ in $\check{G}$, where 
(1) $\check{P}$ is a path between the clusters containing the endpoints of $P$, and 
(2) $\check{P}$ only contains clusters with nontrivial intersection with $P$. 
When $\cC$ is clear from context, we omit the subscript and simply write \EMPH{$\cost(P)$}. 
Notice that $\cC$ is an $(\e, h)$-shortcut partition if, for every path $P$ in $G$, we have $\cost(P) \le \e h \cdot \Ceil{ \frac{\len{P}}{\e\cdot \diam(G)} }$; indeed, for any pair $u,v$ of vertices in $G$, applying this condition for any \emph{shortest} path $P$ between $u$ and $v$ 
yields $\cost(P) \le \e h \cdot \Ceil{ \frac{\len{P}}{\e\cdot \diam(G)} } = 
h \cdot \Ceil{ \frac{\delta_G(u,v)}{\e\cdot \diam(G)} }$, as required.

\medskip
In the rest of this section, we prove the following lemma:

\begin{lemma}
\label{lem:shortcut-cost}
    Let $G$ be a $K_r$-minor-free graph, and let $\Delta$ be a positive number. Then there is a partition $\cC$ of $G$ into connected clusters, such that (1) each cluster has strong diameter at most $4\Delta$,
    and (2) every path $P$ in $G$ with $\len{P} < \Delta/r$ has $\cost(P) \le r^{O(r)}$.
\end{lemma}

We now show that \Cref{lem:shortcut-cost} implies \Cref{thm:shortcut-minor}, which we restate below.

\ShortcutMinor*
\begin{proof} 
Let $\cC$ be the partition guaranteed by \Cref{lem:shortcut-cost} for
$\Delta \coloneqq \frac{\e \cdot \diam(G)}{4}$. Every cluster of $\cC$ has strong diameter at most $4 \Delta = \e \cdot \diam(G)$. To prove that
$\cC$ is an $(\e, h)$-shortcut partition with $h = r^{O(r)}/\e = 2^{O(r \log r)}/\e$,
it suffices to show that $\cost(P) \le r^{O(r)} \cdot \left\lceil \frac{\len{P}}{\e \cdot \diam(G)} \right\rceil$, for an arbitrary path $P$ in $G$. 

We greedily partition $P$ into a sequence of $O \left( \left\lceil \frac{r \len{P}}{\Delta} \right \rceil\right)$ vertex-disjoint subpaths, where each subpath has length at most $\Delta/r$. That is, we can write $P$ as the concatenation $P_1 \circ e_1 \circ P_2 \circ e_2 \ldots \circ Q_\tau$ for some $ \tau = O \left( \left\lceil \frac{r \len{P}}{\Delta} \right \rceil\right)$, such that each $P_i$ has length at most $\Delta/r$. We can upper-bound the cost of $P$: 
$$\cost(P) \le \sum_{1 = 1}^\tau \cost(P_i) + \sum_{i=1}^{\tau-1} \cost(e_i).$$
Each edge has cost at most $1$, and (by \Cref{lem:shortcut-cost}) each subpath $P_i$ has cost at most $r^{O(r)}$. It follows that $\cost(P) \le r^{O(r)} \cdot \left\lceil \frac{\len{P}}{\e \cdot \diam(G)} \right\rceil$, which concludes the proof.
\end{proof}

\subsection{Construction} 

Let $\cT$ be a $(\Delta, \Delta/r, r-1)$-buffered cop decomposition for $G$. 
We partition each supernode $\eta$ into clusters as follows. Fix an arbitrary supernode $\eta$.

Let \EMPH{$N$} be a \EMPH{$\Delta$-net} of the skeleton $T_\eta$ of $\eta$, which is an SSSP tree in $\dom(\eta)$; that is, $N$ is a subset of vertices in $T_\eta$, such that (1) every vertex $v$ in $T_\eta$ satisfies $\dist_{T_\eta}(v, N) = \dist_{\dom(\eta)}(v,N) \le \Delta$, and (2) for every pair of vertices $x_1$ and $x_2$ in $N$, we have $\dist_{T_\eta}(x_1, x_2) 
= \dist_{\dom(\eta)}(x_1,x_2) > \Delta$. (The net $N$ can be constructed greedily.) For each net point in $N$, we initialize a separate cluster. 

We partition the rest of the vertices in $\eta$ based on their closest point in the net $N$. In more detail, consider each vertex $v$ in $\eta$ in increasing order of their distance to $N$.
Find the shortest path $P_v$ from $v$ to the closest point in $N$ (if there are multiple such paths, we fix $P_v$ arbitrarily). 
Let $v'$ be the vertex adjacent to $v$ in $P_v$. Set the cluster of $v$ to be the same as the cluster of $v'$. 
Observe that each cluster has a single net point in $N$, which we refer to as the \emph{center} of the cluster;
the centers of clusters constitute $N$.

\begin{lemma}
\label{C:minor-diameter}
    For each supernode, each of its clusters has strong diameter at most $4\Delta$.
\end{lemma}

\begin{proof}
    Let $\eta$ be an arbitrary supernode and $N$ be the set of cluster centers of $\eta$.
    First, we claim by induction on $\dist_{\eta}(v, N)$ that for every  $v \in \eta$, the cluster $C_v$ that contains $v$ also contains a shortest path from $v$ to $N$.
    For the basis, the claim clearly holds if $v \in N$. For the induction step, suppose that $v \not \in N$.
    Let $P_v$ be the shortest path from $v$ to $N$ that is fixed in our construction, and let $v'$ be the vertex after $v$ in $P_v$. Hence, $v'$ is assigned to $C_v$. Since $\dist_{\eta}(v', N) < \dist_{\eta}(v, N)$, cluster $C_v$ contains a shortest path from $v'$ to $N$, denoted by $Q_v'$, by our induction hypothesis. Hence, the path $(v, v') \circ Q_{v'}$ is a shortest path from $v$ to $N$, which is contained in $C_v$. This completes the induction step.

    Consider an arbitrary cluster $C$ in $\eta$ and any vertex $v \in C$.
    By the supernode radius property, there is a vertex $v_T$ in $T_\eta$ such that $\dist_{\eta}(v, v_T) \le \Delta$; as $N$ is a $\Delta$-net, we have $\dist_{\eta}(v_T, N) \le \Delta$. 
    By the triangle inequality,
    $\dist_{\eta}(v, N) \le 2\Delta$.
    By the above claim, $C$ contains a shortest path from $v$ to $N$, and is thus of length at most $2\Delta$. 
    As $C$ contains a single cluster center,  the diameter of $C$ must be at most $4 \Delta$.
\end{proof}

We now bound the cost of a path $P$. We first prove that ``the highest supernode $\eta$ that $P$ intersects'' is well-defined, then show that $P$ intersects few clusters in $\eta$, and finally give an inductive argument to bound $\cost(P)$.

\begin{claim}
\label{clm:highest-supernode}
    Let $P$ be a path in $G$. 
    $P$ is contained in $\dom(\eta)$ for some supernode $\eta$ that $P$ intersects.
\end{claim}

\begin{proof}
    Let $\mathcal{\hat{T}}$ denote the expansion of the partition tree $\cT$. 
    For every supernode $\eta_1$, the tree decomposition property implies that the set of bags in $\mathcal{\hat{T}}$ containing $\eta_1$ induces a connected subtree of $\mathcal{\hat{T}}$, which we denote $\mathcal{\hat{T}}[\eta_1]$.
    Further, for every pair of supernodes $\eta_1$ and $\eta_2$ that are adjacent in $G$, there is some bag shared by both $\mathcal{\hat{T}}[\eta_1]$ and $\mathcal{\hat{T}}[\eta_2]$; it follows that $\mathcal{\hat{T}}[\eta_1] \cup \mathcal{\hat{T}}[\eta_2]$ is a connected subtree of $\mathcal{\hat{T}}$.
    As $P$ is connected, the set of bags containing \emph{any} supernode that $P$ intersects induces a connected subtree of $\mathcal{\hat{T}}$, which we denote $\mathcal{\hat{T}}[P]$.
    
    Let $\bag{\eta}$ denote the root bag of the subtree $\mathcal{\hat{T}}[P]$, where $\bag{\eta}$ is in one-to-one correspondence with the supernode $\eta$ in $\cT$. Bag $\bag{\eta}$ contains the supernode $\eta$, as well as some supernodes above $\eta$ in $\cT$. Path $P$ does not intersect any supernode above $\eta$, as otherwise $\mathcal{\hat{T}}[P]$ would include a bag above $\bag{\eta}$. Thus, $P$ is in $\dom(\eta)$. Further, $P$ intersects some supernode in $\bag{\eta}$, so $P$ intersects $\eta$.
\end{proof}

\begin{claim}
\label{C:split}
    If $P$ is a path of length less than $\Delta$, and $\eta$ is a supernode such that $P$ is contained in $\dom(\eta)$, then $P$ intersects at most $9 r$ clusters in $\eta$.
\end{claim}

\begin{proof}
    Suppose for contradiction that $P$ satisfies the conditions of the claim, yet it intersects at least $9r + 1$ clusters in $\eta$. 
    By the shortest-path skeleton property, the skeleton $T_\eta$ of $\eta$ is an SSSP tree in $\dom(\eta)$ with at most $r-1$ leaves, thus the vertices of $T_\eta$ can be partitioned into at most $r-1\le r$ shortest paths. Since each cluster in $\eta$ has its center chosen from one of at most $r$ shortest paths,
    $P$ intersects at least $10$ clusters with centers in the same shortest path, denoted by $Q$. 
    Let $C_u$ (respectively, $C_v$) be the first (resp., last) cluster (among those with centers in $Q$) that $P$ intersects, let $u$ (resp., $v$) be an intersection point of $P$ and $C_u$ (resp., $C_v$),  
    and let $c_u$ (resp., $c_v$) be the center of $C_u$ (resp., $C_v$). 
    Since $Q$ is a shortest path in $\dom(\eta)$ that intersects at least 10 centers and as the distance between any two cluster centers is at least $\Delta$, we have $\dist_{\dom(\eta)}(c_u, c_v) \geq 9\Delta$. By the triangle inequality, we have:
    \begin{equation*}
    \begin{aligned}
        ||P|| ~\geq~ \dist_{\dom(\eta)}(u, v) ~\geq~ \dist_{\dom(\eta)}(c_u, c_v) - \underbrace{(\dist_{\dom(\eta)}(c_u, u) + \dist_{\dom(\eta)}(c_v, v))}_{\leq 4\Delta + 4\Delta \quad\text{by \Cref{C:minor-diameter}}} 
        ~\geq~ 9\Delta -  8\Delta ~\ge~ \Delta,
    \end{aligned}
    \end{equation*}
    yielding a contradiction.
\end{proof}

We say that a path $P$ is \EMPH{$k$-constrained} if, for every supernode $\eta$ that $P$ intersects, there are at most $k$ supernodes in the bag $\bag{\eta}$ corresponding to $\eta$. 

\begin{lemma}
\label{lem:cost}
    Let $P$ be a $k$-constrained path with $\len{P} < \Delta/r$. Then $\cost(P) \le (54r)^{k}$.
\end{lemma}

\begin{proof}
    The proof is by induction on $k$. The basis is trivially true, as only the empty path is $0$-constrained.
    We next prove the induction step.
    By \Cref{clm:highest-supernode},
    there is some supernode $\eta$ that $P$ intersects, such that $P$ is in $\dom(\eta)$.
    Choose an arbitrary vertex $v_{\eta}\in P\cap \eta$ and split $P$ at $v_\eta$ into two subpaths, $P'$ and $P''$; $v_{\eta}$ is an endpoint of both  $P'$ and $P''$. We will show $\cost(P') \le 27r \cdot (54r)^{k-1} = \frac{1}{2} (54r)^{k}$, so by symmetry, $\cost(P) = \cost(P') + \cost(P'') \le (54r)^{k}$.

\begin{figure}[ht!]
    \centering
    \includegraphics[width=0.7\textwidth]{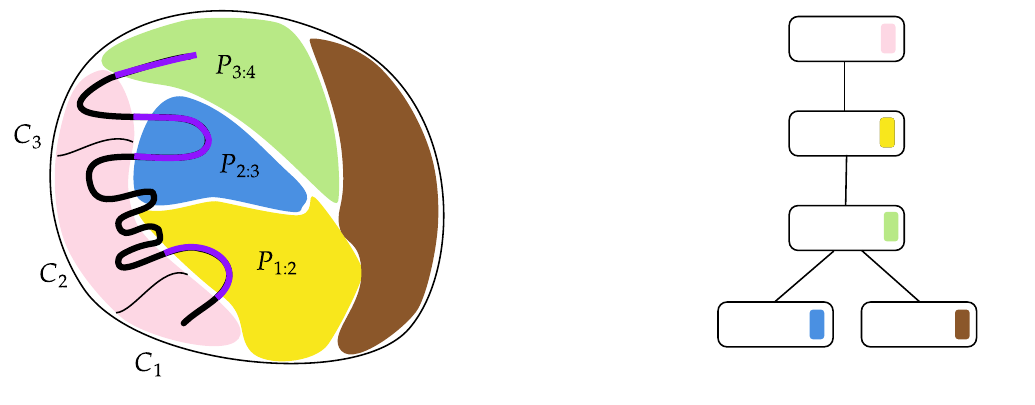}
    \caption{Path $P'$ with subpaths 
    $P_{1:2}$, $P_{2:3} $, and $P_{3:4} $
    highlighted in purple, where $\eta$ is the pink supernode.  }
    \label{fig:subpaths}
\end{figure}

    We partition $P'$ into a sequence of vertex-disjoint subpaths $P_1, P_{1:2}, P_2, P_{2:3}, \ldots$ as follows.
    Let \EMPH{$\cC[\eta]$} denote the set of clusters in $\eta$ that $P'$ intersects.
    Define \EMPH{$C_1$} to be the cluster (in $\cC[\eta]$) that contains $v_\eta$, define \EMPH{$P_1$} to be the maximal prefix of $P'$ that ends in a vertex in $C_1$, and define $\EMPH{$P[C_1:]$} \coloneqq P' \setminus P_1$ to be the suffix of $P'$ starting after $P_1$.
    For all $i \ge 1$ with $C_i \neq \varnothing$, we recursively define (see \Cref{fig:subpaths}):
    \begin{itemize}
        \item Define \EMPH{$C_{i+1}$} to be the first cluster in $\cC[\eta]$ that $P[C_i:]$ intersects. If $P[C_i:]$ intersects no clusters in $\cC[\eta]$, then define $C_{i+1} \coloneqq \varnothing$.
        
        \item Define \EMPH{$P_{i:i+1}$} to be the maximal prefix of $P[C_i:]$ that contains no vertex in $C_{i + 1}$. 
        If $C_{i+1} = \varnothing$, then set $P_{i:i+1} \coloneqq P[C_i:]$. Notice that $P_{i:i+1}$ may be the empty path if the first vertex on $P[C_i:]$ is in~$C_{i+1}$.
        
        \item Define \EMPH{$P_{i+1}$} to be the maximal subpath of $P[C_i:]$ with both endpoints in $C_{i+1}$. Notice that $P_{i+1}$ starts immediately after $P_{i:i+1}$.
        
        \item Define \EMPH{$P[C_{i+1}:]$} to be the suffix of $P[C_{i}:]$ that starts after $P_{i+1}$.
        Notice that $P[C_{i+1}:]$ contains no vertices in $C_{i+1}$.
    \end{itemize}
    By \Cref{C:split}, $\cC[\eta]$ contains at most $9r$ clusters%
    \footnote{\Cref{C:split} is stronger than what we need here}; thus, there are at most $9r$ subpaths $P_i$ and $9r$ subpaths $P_{i : i+1}$ defined by the above procedure. There are at most $18r$ edges in $P'$ that connect the subpaths.
    The cost of $P'$ is bounded by the sum of costs of the subpaths as well as the edges between the subpaths.
    \begin{itemize}
        \item Each edge has cost at most 1.
        \item Each subpath $P_i$ has cost 0, as either $P_i$ is empty or its endpoints are in the same cluster.
        \item As we argue next, \emph{each subpath $P_{i:i+1}$ has cost at most $(54r)^{k-1}$}.

        Observe that every supernode $\eta'$ that $P'$ intersects has $\eta$ in its bag $\bag{{\eta'}}$. Indeed, if $\bag{\eta'}$ did not contain $\eta$, then $\eta'$ and $\eta$ would not be adjacent by definition; as $\eta$ is above $\eta'$ in $\cT$ (by definition of $\eta$), the supernode buffer property implies that $\len{P} \ge \delta_{\dom(\eta)}(\eta',\eta) \ge \Delta/r$, a contradiction.
        Further, notice that $P_{i : i+1}$ does not intersect $\eta$. 
        Thus, as $P'$ is $k$-constrained, each subpath $P_{i:i+1}$ is $(k-1)$-constrained. 
        The inductive hypothesis implies that $\cost(P_{i:i+1}) \le (54r)^{k-1}$, as argued.
    \end{itemize}
    Since $k \ge 1$, we conclude that
    \begin{align*}
      \cost(P') ~\le~ 18r \cdot 1 + 9r \cdot 0 + 9r \cdot (54r)^{k-1} 
      ~\le~ 27r \cdot (54r)^{k-1}.
    \end{align*}
This proves the lemma.
\end{proof}

\noindent Noting that every path is $(r-1)$-constrained (as every bag contains at most $r-1$ supernodes), \Cref{C:minor-diameter} and \Cref{lem:cost} prove \Cref{thm:shortcut-minor}.


\section{Other Applications}\label{sec:other-apps}

In this section, we describe several applications of our shortcut partition mentioned in the introduction.
We will start with two direct applications (\Cref{sec:direct}) and then proceed to 
the application on the
embedding of apex-minor-free graphs (\Cref{sec:embedapex}).

\subsection{Direct Applications} \label{sec:direct}

\paragraph{Tree cover.} 
The authors of \cite{CCLMST23} provided a construction of $(1+\eps)$-tree cover for minor-free graphs via shortcut partition. Their definition of $(\e, h)$-shortcut partition is slightly different than our~\Cref{def:shortcut-partition}; it is strictly weaker. Recall that the low-hop property of \Cref{def:shortcut-partition} guarantees:
\begin{quote}
    For any pair of vertices $u$ and $v$ in $G$, there exists some shortest path $\pi$ in $G$ between $u$ and $v$, and a path $\check{\pi}$ in the cluster graph $\check{G}$ such that (in addition to other properties) $\check{\pi}$ only contains clusters that intersect $\pi$.
\end{quote}
The definition of~\cite{CCLMST23} (Definition 2.1) instead only guarantees: 
\begin{quote}
    For any pair of vertices $u$ and $v$ in $G$, there exists some path $\pi'$ in $G$ between $u$ and $v$ \ul{with $\len{\pi'} \le (1+\e)\delta_G(u,v)$}, and a path $\check{\pi}$ in the cluster graph $\check{G}$ such that (in addition to other properties) $\check{\pi}$ only contains clusters that intersect $\pi'$.
\end{quote}
This is a weaker guarantee. 
Also, as mentioned already, the definition of \cite{CCLMST23}
states that the hop-length of $\check{\pi}$ is at most $h$, regardless of $\dist_G(u,v)$, while our definition takes $\dist_G(u,v)$ into account,
allowing smaller hop-lengths for smaller distances.
Consequently, the shortcut partition we construct in~\Cref{thm:shortcut-minor} can be used in their construction of tree cover.

The tree cover construction of~\cite{CCLMST23} consists of two steps.%
\footnote{The more efficient tree cover construction for \emph{planar graphs} in \cite{CCLMST23} does not follow the two-step framework; instead, the authors exploited planarity to get a better (and indeed polynomial) dependency on $\eps$ in the size of the cover.} 
The first step 
is a reduction from a tree cover with multiplicative distortion $(1+\eps)$ to a tree cover with additive distortion $+\eps\Delta$, where $\Delta$ is the diameter, with a loss of a $O(\log(1/\eps))$ factor to the cover size. 
In the second step,
it is shown that an $(\eps,h)$-shortcut partition for minor-free graphs implies a tree cover of size $2^{O(h)}$ with additive distortion $+\eps \Delta$. 
Their result can be summarized as follows.

\begin{lemma}[Lemma 1.7 and Theorem 1.8 in~\cite{CCLMST23}]
\label{lem:cclmst-treecover}
    Let $G$ be a minor-free graph, and $\e \in (0,1)$. 
    If every subgraph of $G$ admits an $(\e, h)$-shortcut partition, then $G$ has a $(1+\eps)$-tree cover of size $2^{O(h)}$. 
\end{lemma}

\noindent
By \Cref{thm:shortcut-minor}, any $K_r$-minor-free graph has an $(\e, r^{O(r)}/\e)$-shortcut partition. 
This proves~\Cref{thm:tree-cover}.

\paragraph{Distance oracle.} 
As discussed in \Cref{subsec:app-intro}, we construct our distance oracle from our tree cover (\Cref{thm:tree-cover}): given a query pair $(u,v)$, query the distance $d_T(u,v)$ in $T$,
for each tree $T$ in the tree cover $\mathcal{T}$,
and return $\min_{T\in \mathcal{T}}d_T(u,v)$. Distance query in a tree is reduced to a lowest common ancestor (LCA) query. 
In the RAM model, there are LCA data structures \cite{HT84,BF00} with $O(n)$ space and $O(1)$ query time. 
In the pointer machine model, this can be carried out
with $O(n)$ space and $O(\log\log n)$ query time \cite{VanLeeuwen76}. 
\Cref{thm:app-distance-oracle} now follows.

\subsection{Embedding of apex-minor-free graphs} \label{sec:embedapex}

The authors of \cite{CCLMST23} show that any \emph{planar} graph $G$ with diameter $\Delta$ can be embedded into a graph of treewidth $O(\e^{-4})$ with distortion $+\e\Delta$. 
Their argument uses three properties of planar graphs, which carries over to any minor-free graph with these properties (\Cref{lem:cclmst-treewidth}). Loosely speaking, they show that if (P1) $G$ has an $(\e, h)$-shortcut partition, (P2) $G$ has an $+\e\Delta$ forest cover $\cF$ for $G$ that ``interacts nicely'' with the shortcut partition (\Cref{lem:embedding-properties}), and (P3) $G$ has the \emph{local-treewidth} property,
then $G$ can be embedded into a graph of treewidth $O(h \cdot \abs{\cF})$ with distortion $+\e \Delta$.

\Cref{thm:shortcut-minor} gives us a shortcut partition for minor-free graphs and hence (P1). Apex-minor-free graphs have the local treewidth property~\cite{Eppstein00,DH04,DH04B} (see \Cref{lm:diameter-treewidth} below) and hence satisfy (P3). The main goal of this section is to show (P2) (by proving~\Cref{lem:embedding-properties}), and we do so by applying the framework of~\cite{CCLMST23} to our shortcut partition to construct an appropriate forest cover.

\begin{lemma}[Diameter-treewidth property~\cite{DH04B}]\label{lm:diameter-treewidth} Let $G$ be a graph excluding a fixed apex graph as a minor. Let $D$ be its (unweighted) diameter. Then $\tw(G) = O(D)$. 
\end{lemma}
We note that the big-O in \Cref{lm:diameter-treewidth} hides the dependency on the size of the minor; it is the Robertson-Seymour constant. As a result, the dependency on the minor of our \Cref{thm:add-apex-minor} also has a Robertson-Seymour constant. 
  
\medskip
We recall the construction of the~\cite{CCLMST23} tree cover for completeness. 
A forest $F$ is \EMPH{dominating} if $d_F(u,v)\geq d_G(u,v)$ for every two vertices $u,v\in V(G)$. In the definition of the tree cover for $G$, we allow \emph{Steiner} vertices, i.e., which do not belong to $G$, in a tree. Here the forest we construct will contain no Steiner vertices, meaning that $V(F)\subseteq V(G)$. In this case, we say that $F$ is \EMPH{Steiner-free}. 
Let $\mathcal{C}$ be a clustering of $G$. 
Let $\check{G}$ be the \emph{cluster graph} obtained from $G$ by contracting clusters in $\mathcal{C}$; $\check{G}$ is a simple unweighted graph. 
Let $\check{F}$ be a forest, subgraph of $\check{G}$ such that every tree in $\check{F}$ is rooted at some node. We define the \EMPH{star expansion} of $\check{F}$ to be a forest $F$ of $G$ obtained by applying the following to each tree $\check{T}\in \check{F}$ (see \Cref{fig:forest-correspodence}):
\begin{quote}
    Let $V_T$ denote the set of vertices (of $G$) that belong to clusters in $\check{T}$. Choose an arbitrary vertex $r$ in the cluster that is the root of $\check{T}$. Let   $T$ be a star rooted at $r$ connected to every vertex in $V_T$. We assign each edge $(r,u)$ of $T$ a weight $d_G(r,u)$. We then add $T$ to $F$.
\end{quote}

\begin{figure}[h!]
    \centering
    \includegraphics[width=0.9\textwidth]{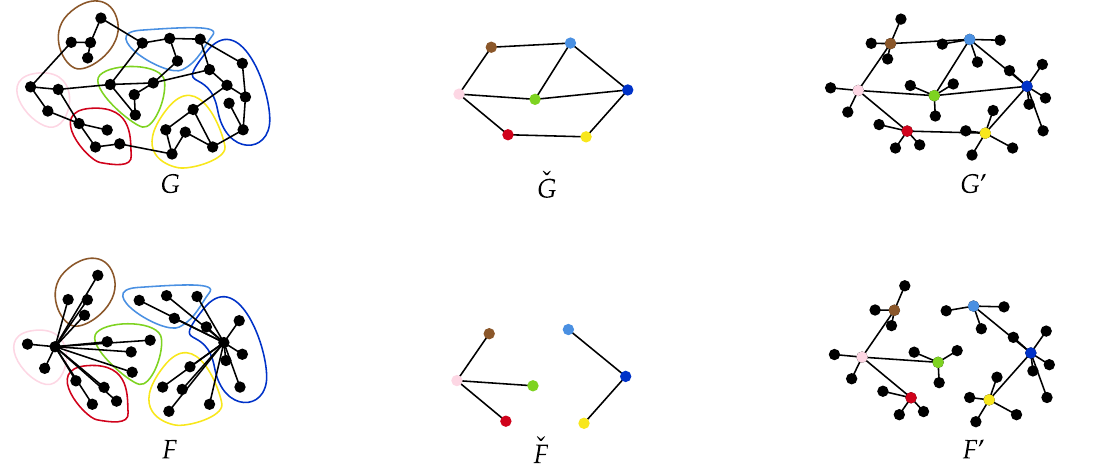}
    \caption{An illustration of $G$, $\check{G}$, $G'$, $F$, $\check{F}$, and $F'$. }
    \label{fig:forest-correspodence}
\end{figure}

Let $\check{\cF}$ be a collection of rooted forests of $\check{G}$, each of which is a subgraph of $\check{G}$, called a \emph{spanning forest cover}. The \emph{star expansion} of  $\check{\cF}$ is a collection of rooted forests, denoted by $\mathcal{F}$, obtained by taking star expansion of each rooted forest $\check{F}$ in $\check{\cF}$.  We note that a forest in $\check{\cF}$ might not be a subgraph of $G$, but it is Steiner-free. (That is, each forest might contain Steiner edges---edges not in $G$---but not Steiner vertices.) The following lemma was proven in~\cite{CCLMST23}.

\begin{lemma}[Adapted from Theorem 2.2 and Theorem 2.5 in  \cite{CCLMST23}] 
\label{lem:refined-treecover} 
Let $G$ be an edge-weighted minor-free graph with diameter $\Delta$, and let $\e \in (0,1)$. 
Suppose that $G$ has an $(\e, h(\e))$-shortcut partition $\mathcal{C}$. 
Let $\check{G}$ be the cluster graph obtained from $G$ by contracting clusters in $\mathcal{C}$.  
Then there exists a spanning forest cover $\check{\cF}$ of $\check{G}$ such that its star expansion $\cF$ satisfies the following properties: 

\begin{itemize}
       \item \textnormal{[Root preservation.]} For every pair of vertices $u$ and $v$ in $G$, there is a tree $T$ in some forest of $\mathcal{F}$ such that (1) $\dist_T(u,v) \le \dist_G(u,v) + \e \Delta$, and (2) a shortest path from $u$ to $v$ in $T$ passes through the root of $T$.
       \item  \textnormal{[Size.]} $\cF$ contains $2^{O(h(\e))}$ forests. 
    \end{itemize}
\end{lemma}
We remark that the second half of the root preservation property was not explicitly stated in~\cite{CCLMST23}; however, it is immediate from the fact that $\cF$ is a set of star forests.
 
\medskip
We now show the following structural result for apex-minor-free graphs (see \Cref{fig:forest-correspodence} for illustration).

\begin{lemma}
\label{lem:embedding-properties}
    Let $G$ be an edge-weighted graph with diameter $\Delta$ excluding a fixed apex graph as a minor, and let $\e \in (0,1)$. 
    Then there is a partition $\cC$ of the vertices of $G$ into clusters with strong diameter $\e \Delta$, and a set $\mathcal{F}$ of $2^{O(1/\e)}$ forests with the same vertex set as $G$, that satisfy the following properties:
    \begin{itemize}
        \item \textnormal{[Low-hop.]} For every pair of vertices $u$ and $v$, there is a path between $u$ and $v$ that intersects at most $h$ clusters, for some $h = O(1/\e)$.
        \item \textnormal{[Root preservation.]} For every pair of vertices $u$ and $v$ in $G$, there is a tree $T$ in some forest of $\mathcal{F}$ such that (1) $\dist_T(u,v) \le \dist_G(u,v) + \e \Delta$, and (2) a shortest path from $u$ to $v$ in $T$ passes through the root of $T$.
    \end{itemize}
    For each cluster $C$ in $\cC$, choose an arbitrary vertex $v_C$ to be the \emph{center vertex}, and define the \emph{star} $S_C$ to be a star connecting $v_C$ to every other point in $C$. 
    Define \EMPH{$G'$} to be the graph obtained by replacing every supernode $C$ in cluster graph $\check{G}$ with the star $S_C$: every edge between two clusters in $\check{G}$ is replaced with an edge in $G'$ between the centers of the two clusters. Notice $G'$ has the same vertex set as $G$.
    
    \begin{itemize}
        \item \textnormal{[Contracted treewidth.]} 
        Graph $G'$ has treewidth $O(h)$.
        \item \textnormal{[Forest correspondence.]} For every forest $F$ in $\mathcal{F}$, there is a corresponding spanning forest $F'$ (a subgraph of $G'$) such that: For every tree $T$ in $F$, there is a tree $T'$ in $F'$ such that $V(T) = V(T')$ and $\mathrm{root}(T) = \mathrm{root}(T')$.
    \end{itemize}
\end{lemma}

\begin{proof}
    By \Cref{thm:shortcut-minor}, there is a partition $\cC$ of $G$ that is an $(\e, h)$-shortcut partition, for $h = O(1/\e)$. By \Cref{lem:refined-treecover}, we can construct a spanning forest cover $\check{\mathcal{F}}$ of $\check{G}$ and its star expansion $\cF$ satisfying the root preservation property. Furthermore, $\cF$ has $2^{O(1/\eps)}$ forests. We now show the other three properties of the theorem.

    \medskip\noindent
    \textbf{[Low-hop.]} The low-hop property follows directly from the statement of \Cref{thm:shortcut-minor}.

    \medskip\noindent
    \textbf{[Contracted treewidth.]} Let $\check{G}$ denote the graph obtained by contracting every cluster in $\cC$ into a supernode. 
    By the low-hop property, the (unweighted) diameter of $\check{G}$ is at most $h$. Graph $\check{G}$ excludes the same minors as $G$, so \Cref{lm:diameter-treewidth} implies that $\check{G}$ has treewidth $O(h)$. 
    Now notice that we can obtain $G'$ from $\check{G}$ by creating new (degree-1) vertices and adding an edge between each new vertex and a supernode in $\check{G}$.
    We construct a tree decomposition for $G'$, starting from the tree decomposition for $\check{G}$: 
    For each new vertex $v$ attached to an existing supernode $u$, create a bag containing $u$ and $v$, and add it as a child to an arbitrary bag containing $u$ in the tree decomposition of $\check{G}$. 
    This procedure does not change the treewidth; it is still $O(h)$.

    \medskip\noindent
    \textbf{[Forest correspondence.]} By definition, each forest $F\in \cF$ is a star expansion of a forest $\check{F} \in \check{\cF}$. To get the forest correspondence property, we simply transform  $\check{F}$ into a forest $F'$ on $G'$ in the natural way: For each tree $\check{T}$ in $\check{F}$, replace every vertex $C$ in $\check{T}$ with the corresponding star $S_C$ in $G'$, and replace every edge in $\check{T}$ with the corresponding edge between star centers in $G'$. We claim that $F'$ is a forest. Indeed, this transformation maps each tree $\check{T}$ to a tree $T'$ in $G'$, because $T'$ is connected and has one more vertex than it has edges. Recall that for each tree $\check{T}\in \check{F}$, there is a corresponding tree $T\in F$, which is obtained by star expansion. By construction, $T$ and $T'$ has the same vertex set. We then can set the root of $T'$ to be the same as the root of $T$.

    Further, the trees $T'$ are vertex disjoint because the clusters $\cC$ are vertex disjoint, and there are no edges between the trees $T'$. Thus, $F'$ is a spanning forest of $G'$.
\end{proof}

The following reduction is implicit in~\cite{CCLMST23}.%
\footnote{There are small differences in the phrasing of \cite{CCLMST23}. 
(1) They do not explicitly state the \emph{contracted treewidth} property; rather, they prove it as a lemma using properties of planar graphs. 
(2) They do not explicitly state the \emph{forest correspondence} property; rather, they state a different condition (the ``disjoint cluster'' condition), that they then use to prove the forest correspondence property in a lemma. The ``disjoint cluster'' condition is used only to prove the forest correspondence property.}

\begin{lemma}[\cite{CCLMST23}, Section 7]
\label{lem:cclmst-treewidth}
    Let $G$ be an edge-weighted minor-free graph with diameter $\Delta$, and let $\e$ be a number in $(0,1)$. 
    Suppose there is a partition $\cC$ of $G$ into clusters of strong diameter $\e \Delta$, together with a set of forests $\cF$, such that $\cC$ and $\cF$ satisfy  
    the low-hop property with parameter $h$, the root preservation property, 
    the contracted treewidth property, and the forest correspondence property. 
    Then $G$ can be embedded deterministically into a graph with treewidth $O(h \cdot \abs{\cF})$ and distortion $+O(\e \Delta)$.
\end{lemma}

Combining \Cref{lem:embedding-properties} and \Cref{lem:cclmst-treewidth} proves \Cref{thm:add-apex-minor}.


\paragraph{Acknowledgement.~} Hung Le and Cuong Than are supported by the NSF CAREER Award No. CCF-2237288 and an NSF Grant No. CCF-2121952.

\small
\bibliographystyle{alphaurl}
\bibliography{main}

\end{document}